\documentclass[12pt]{article}

\usepackage[T1]{fontenc}
\usepackage[latin1]{inputenc}

\usepackage{bbm}
\usepackage{mathrsfs}
\usepackage{latexsym}
\usepackage{lmodern}

\usepackage{amsmath,amsfonts,amssymb,amstext,amsthm} 
\usepackage{mathtools}

\usepackage[usenames,dvipsnames,table]{xcolor}
\usepackage{graphicx}

\usepackage{wrapfig}
\usepackage{booktabs} 
\usepackage{longtable}
\usepackage{float}

\usepackage{enumerate}
\usepackage{enumitem}
\usepackage{xspace}
\usepackage{comment}
\usepackage{ifthen}
\usepackage[noadjust,nospace]{cite} 
\usepackage{ellipsis}
\usepackage{fixltx2e}
\usepackage{url}
\usepackage{scrtime}

\usepackage{microtype} 

\usepackage[ruled,longend,vlined,linesnumbered]{algorithm2e}

\usepackage[draft,author=]{fixme}
\fxsetup{theme=color}

\newcommand{\FPT}{\ensuremath{\text{FPT}}\xspace}

{\bfseries \upshape}{\itshape}
{\bfseries \upshape}{\itshape}

\newcommand{\name}[1]{\textsc{#1}}

\newcommand{\NP}{\ensuremath{\text{NP}}\xspace}

\renewcommand{\leq}{\leqslant}

\renewcommand{\geq}{\geqslant}

\newcommand{\widthm}[1]{\ensuremath{\mathop\mathbf{#1}}\xspace}
\newcommand{\tw}{\widthm{tw}}
\newcommand{\pw}{\widthm{pw}}
\newcommand{\td}{\widthm{td}}

\newcommand{\treedecomp}{\ensuremath{\mathcal{T}}}

\newcommand{\clos}{\ensuremath{\text{clos}}}

\newcommand{\height}[2]{\ensuremath{\text{height}_{#1}(#2)}}
\newcommand{\rootof}[2]{\ensuremath{\text{root}_{#1}(#2)}}
\newcommand{\anc}[2]{\ensuremath{\text{Anc}_{#1}(#2)}}
\newcommand{\des}[2]{\ensuremath{\text{Des}_{#1}(#2)}}
\newcommand{\child}[2]{\ensuremath{\text{Children}_{#1}(#2)}}

\newcommand{\intro}{\ensuremath{\mathit{intro}}}
\newcommand{\forget}{\ensuremath{\mathit{forget}}}
\newcommand{\join}{\ensuremath{\mathit{join}}}

\newtheorem{lemma}{Lemma}
\newtheorem{theorem}{Theorem}
\newtheorem{corollary}{Corollary}

\newtheorem{proposition}{Proposition}

\newtheorem*{prob}{Problem}

\theoremstyle{definition}

\newtheorem{definition}{Definition}

\newcommand*{\ie}{i.e.\@\xspace}
\newcommand*{\cf}{cf.\@\xspace}

\makeatletter
\newcommand*{\etc}{%
    \@ifnextchar{.}%
        {etc}%
        {etc.\@\xspace}%
}
\makeatother
\newif\ifshort

\usepackage{framed}
\usepackage{ctable}

\def\compxy(#1){\mathcal{U}_{#1}}

\def\Nesetril{Ne\v{s}et\v{r}il\xspace}

\shortfalse

\usepackage{authblk}

\title{A Faster Parameterized Algorithm for Treedepth\thanks{Research funded by 
DFG-Project RO 927/13-1 ``Pragmatic Parameterized Algorithms''.}}

\author{\small Felix Reidl}
\author{\small Peter Rossmanith}
\author{\small Fernando S\'{a}nchez Villaamil}
\author{\small Somnath Sikdar}

\affil{\small Theoretical Computer Science, Department of Computer Science,
  
RWTH Aachen University, Aachen, Germany, 

\texttt{\{reidl,rossmani,fernando.sanchez,sikdar\}@cs.rwth-aachen.de}.
}

\date{\today}

\begin{document}
\maketitle

\begin{abstract}
The width measure \emph{treedepth}, also known as vertex ranking,
centered coloring and elimination tree height, is a well-established
notion which has recently seen a resurgence of interest. We present an
algorithm which---given as input an $n$-vertex graph, a tree
decomposition of the graph of width $w$, and an integer~$t$---decides
\textsc{Treedepth}, \ie~whether the treedepth of the graph is at
most~$t$, in time $2^{O(wt)} \cdot n$. If necessary, a witness
structure for the treedepth can be constructed in the same running
time. In conjunction with previous results we provide a simple
algorithm and a fast algorithm which decide treedepth in time
$2^{2^{O(t)}} \cdot n$ and $2^{O(t^2)} \cdot n$, respectively, which
do not require a tree decomposition as part of their input. The former
answers an open question posed by Ossona de Mendez and \Nesetril\ as
to whether deciding \textsc{Treedepth} admits an algorithm with a
linear running time (for every fixed~$t$) that does not rely on
Courcelle's Theorem or other heavy machinery. For chordal graphs we
can prove a running time of $2^{O(t \log t)}\cdot n$ for the same
algorithm.


\end{abstract}
\newpage

\section{Introduction}\label{sec:Introduction}
The notion of treedepth has been introduced several times 
in the literature under several different names. 
It seems that this was first formally studied by Pothen 
who used the term \emph{minimum elimination tree}~\cite{Pot88};
Katchalski et al.~\cite{KMS95} studied the same notion 
under the name of \emph{ordered colorings};
Bodlaender et al. in~\cite{BDJ98} used the term \emph{vertex ranking}.
More recently, Ossona de Mendez and Ne\v{s}et\v{r}il brought 
the same concept to the limelight in the guise of \emph{treedepth}
in their book \textit{Sparsity}~\cite{NOdM12}. 

As is to be expected, there are several equivalent definitions 
of this term. One of the most intuitive characterizations of treedepth 
is via the degeneracy of the graph: a graph class has bounded treedepth if 
and only if the class is degenerate and there exists a constant~$l$ (that depends on 
the class) such that no graph from the class has an induced path of length at least~$l$ 
(Theorem~13.3 in~\cite{NOdM12}). That is, the condition that a graph has 
bounded treedepth imposes a slightly stronger restriction than just bounding 
the degeneracy of the graph: it also implies that the graph has no long 
induced paths. A particularly simple definition of treedepth is via the notion of 
vertex rankings. A \emph{$t$-ranking} of a graph $G = (V,E)$ is a vertex 
coloring $c \colon V \to \{1, \ldots, t\}$ such that for any two vertices 
of the same color, any path connecting them has a vertex with a higher color. 
The minimum value of $t$ for which such a coloring exists is the 
\emph{treedepth} or the \emph{vertex ranking number} of the graph. 
We denote the treedepth of a graph $G$ by $\td(G)$.
The vertex ranking number finds applications in sparse matrix factorization~\cite{DR83,Liu90,KU13} 
and VLSI layout problems~\cite{Lei80}. This notion also has important connections 
to the structure of sparse graphs. As Ossona de Mendez and Ne\v{s}et\v{r}il 
show in~\cite{NOdM08a}, a very general class of sparse graphs, the so-called 
graphs of bounded expansion, can be decomposed into pieces of bounded treedepth. 

Formally, the \name{Treedepth} problem is to decide, given a graph~$G$ and an integer~$t$, 
whether $G$ has treedepth at most~$t$.  This decision problem is \NP-complete 
even on co-bipartite graphs as shown by Pothen~\cite{Pot88} and later by Bodlaender
et al.~\cite{BDJ98}. On trees, the problem can be decided in linear time~\cite{Sch89}. 
Deogun et al.~\cite{DKKM94} showed that \name{Treedepth} can be computed in 
polynomial time on the following graph classes: permutation, circular permutation, 
interval, circular-arc, trapezoid graphs and also on cocomparability graphs
of bounded dimension. It is, however, NP-hard on chordal graphs~\cite{DN06}.
On general graphs, the problem does not admit good 
approximation algorithms: the best-known approximation algorithm is due to 
Bodlaender et al.~\cite{BGHK95} and has performance 
ratio $O(\log^2 n)$, where~$n$ is the number of vertices in the graph.
The best-known exact algorithm for this problem is due to Fomin, Giannopoulou and 
Pilipzcuk~\cite{FGP13} and runs in time $O^{*}(1.9602^n)$. For 
practical applications, several simple heuristics exist.    
One such heuristic is to find a balanced vertex separator, assign each vertex 
of the separator a distinct color and then recurse on the remaining components. 
This method shows that $n$-vertex planar graphs have a treedepth of $O(\sqrt{n})$. 
There are several good heuristics for obtaining balanced separators and some of 
the most practically useful ones rely on spectral techniques (see, 
for instance~\cite{PSL90,ST96}). 

Concerning parameterized complexity, it is ``easy'' to see that \name{Treedepth} 
is fixed-parameter tractable with the solution size as parameter. This follows 
from the fact that graphs of bounded treedepth are minor-closed and hence, 
by the celebrated Graph Minors Theorem of Robertson and Seymour, are 
characterized by a finite set of forbidden minors. Again, by Robertson and 
Seymour~\cite{RS95c}, one can test whether $H$ is a minor of a graph $G$ in 
time $O(f(h) \cdot n^3)$, where $n$ is the number of vertices in $G$, 
$h$ is the number of vertices in $H$ and $f$ is some recursive function. 
Therefore, for every fixed~$t$, one can decide whether a graph contains 
as minor a member of the (finite) set that characterizes graphs of treedepth~$t$
in time $O(g(t) \cdot n^3)$, for some recursive function~$g$. In their 
textbook \textit{Sparsity}, Ossona de Mendez and Ne\v{s}et\v{r}il present in~\cite{NOdM12} an 
algorithm that relies on Courcelle's theorem to show that for every 
fixed~$t$, the \name{Treedepth} problem can be decided in linear
time. They also present the following as an open problem:

\begin{prob}
  Is there a simple linear time algorithm to check $\td(G) \leq t$ for
  fixed $t$? 
\end{prob}

Bodlaender et al.\ in~\cite{BDJ98} provide a dynamic programming 
algorithm that takes as input, a graph~$G$ and a tree decomposition 
of~$G$ of width~$w$, and decides whether~$G$ has treedepth at
most~$t$ in time\footnote{We point out that the 
running time analysis in~\cite{BDJ98} simply states that the algorithm
runs in \emph{polynomial} time for a fixed~$t$ and~$w$. However, it is not difficult 
to restate the running time to include~$t$ and~$w$ as parameters, which is what we have done. 
From a personal communication~\cite{BK14}, it seems that the running time can be 
improved to $2^{O(w^2t)} n$.} $2^{O(w^2 t)}\cdot n^2$.
In this paper we present a linear time algorithm
that decides whether $\td(G) \leq t$ in time $2^{O(wt)} \cdot n$,
improving both the dependence on $w$ and $n$. If indeed $\td(G) \leq
t$, then the algorithm also constructs a treedepth decomposition
within this time. That a better dynamic programming algorithm can be
achieved using treedepth leads us to believe that representing the
ranking of the vertices as a tree might be algorithmically helpful in
other cases.

We can then, by using previous known characteristics of treedepth,
easily extend this result to get the following two algorithms:

\begin{itemize}
\item A simple algorithm which runs in time $2^{2^{O(t)}} \cdot n$.
\item A fast algorithm which runs in time $2^{O(t^2)} \cdot n$ using a
  $5$-approximation for treewidth by Bodlaender
  et.~al.~\cite{BDDFLP13}.
\end{itemize}

We would like to point out that the second algorithm has a lower
exponential dependency on the treedepth than the best-known algorithm
to decide the treewidth of a graph, which is $2^{O(w^3)} \cdot n$, has
on the treewidth.
    

\section{Preliminaries}\label{sec:Preliminaries}
We use standard graph-theoretic notation (see~\cite{Die10} for any
undefined terminology). All our graphs are finite and simple. Given a
graph $G$, we use $V(G)$ and $E(G)$ to denote its vertex and edge
sets. For convenience we assume that $V(G)$ is a totally ordered set,
and use $uv$ instead of $\{u,v\}$ to denote the edges of $G$. For~$X
\subseteq V(G)$, we let~$G[X]$ denote the subgraph of $G$ induced by
$X$. We need the notion of edge contraction. Given an edge~$e = uv$ of
a graph~$G$, we let~$G/e$ denote the graph obtained from~$G$ by
\emph{contracting} the edge~$e$, which amounts to deleting the
endpoints of~$e$, introducing a new vertex~$w_{uv}$, and making it
adjacent to all vertices in $(N(u) \cup N(v)) \setminus \{u,v\} $. For
an edge $e = uv$, by \emph{contracting $v$ into $u$}, we mean
contracting $e$ and renaming the vertex~$w_{uv}$ by~$u$. For a
function $f \colon X \to Y$ and a set $X' \subseteq X$ we will define
applying the function on such a set to be $f(X') = \{f(x) \,|\, x \in
X'\}$.

\subsection{Forests}
\label{sec:forests}

We will work extensively on trees and forests. In this context, 
a \emph{rooted tree} is a tree with a specially designated node 
known as the \emph{root}.
Let $T$ be a rooted tree with root~$r$ and let $x \in V(T)$. Then an 
\emph{ancestor} of $x$ is any node (other than itself) on the path 
from $r$ to $x$. Similarly a \emph{descendant} of $x$ is any node 
(other than itself) on a path from $x$ to a leaf of $T$. In particular, 
$x$ is neither an ancestor nor a descendant of itself.  

A \emph{rooted forest} is a disjoint union of rooted trees. Whenever we
refer to a forest we will mean a rooted forest. For a node~$x$ in a
tree~$T$ of a forest, the \emph{depth} of~$x$ in the forest is the
number of vertices in the path from the root of~$T$ to~$x$ (thus
the depth of the root is one). The
\emph{height of a forest} is the maximum depth of a node of the
forest.  The \emph{closure} $\clos(F)$ of a forest~$F$ is the graph
with node set $\bigcup_{T \in F} V(T)$ and edge set $\{xy \,|\,
\text{$x$ is an ancestor of $y$ in $F$}\}$. Furthermore we will need
the notions of a \emph{subtree} and the \emph{height} of a
node.

\begin{definition}[Subtree rooted at a node]
  Let $x$ be a node of a tree $T$ and let $S$ be all the descendants
  of $x$ in $T$. The \emph{subtree of $T$ rooted at $x$}, denoted by $T_x$, is
  the subtree of $T$ induced by the node set $S \cup \{x\}$ with root~$x$.
\end{definition}

\begin{definition}[Subtree rooted at a node with child selection]
  Let $x$ be a node of a tree $T$, let $C$ be a set of children of $x$
  in $T$ and let $S$ be all descendants of nodes of $C$ in $T$. The
  tree denoted by $T_x^C$, is the subtree of $T$ induced by the node
  set $S \cup C \cup \{x\}$ with root~$x$.
\end{definition}

\begin{definition}[Height of a node]
  Let $x$ be a node of a tree $T$ and let $T_x$ be the subtree of $T$ 
  rooted at $x$. Then the \emph{height} of $x$ is the height of $T_x$.
\end{definition}

\subsection{Treedepth and Treewidth}

Our main algorithm is a dynamic programming algorithm that works on a
tree decomposition.

\begin{definition}[Treewidth]
  Given a graph~$G=(V,E)$, a \emph{tree decomposition of $G$} is an
  ordered pair $(T, \{ W_x \mid x \in V(T) \})$, where~$T$ is a tree
  and $\{W_x \mid x \in V(T)\}$ is a collection of subsets of~$V(G)$ 
  such that the following hold:
  \begin{enumerate}
  \item $\bigcup_{x \in V(T)} W_x = V(G)$;
  \item for every edge~$e = uv$ in~$G$, there exists~$x \in V(T)$ such
    that~$u,v \in W_x$;
  \item for each vertex~$u \in V(G)$, the set of nodes~$x \in V(T)$ 
    such that $u \in W_x$ induces a subtree of $T$. 
  \end{enumerate}
  We call the vertices of $T$ \emph{nodes} (as opposed to ``vertices'' of $G$). 
  The vertex sets $W_x$ are usually called \emph{bags}. 
  The \emph{width} of a tree decomposition is the size of the largest bag minus one. The
  \emph{treewidth} of~$G$, denoted~$\tw(G)$, is the smallest width of
  a tree decomposition of~$G$.
\end{definition}
In the definition above, if we restrict $T$ to being a path, we obtain
the well-known notions of a \emph{path decomposition} and
\emph{pathwidth}. We let $\pw(G)$ denote the pathwidth of $G$.
Let $(T, \{ W_x \mid x \in V(T) \})$ be a tree-decomposition; let 
$x \in V(T)$ and, let $S$ be the set of  descendants 
of $x$. Then  we define 
$V(\mathcal T_{W_x}) := \bigcup_{y \in S \cup \{x\}} W_y$.

We will only work on \emph{nice tree decompositions}, 
which are tree decompositions with the following characteristics:
\begin{itemize}
\item Every node has either zero, one, or two children. 
\item Bags associated with leaf nodes contain a single vertex.
\item If $x$ is a node of $T$ with a single child $x'$ and if $X$ and $X'$ 
      are the bags assigned to these nodes, then either $|X \setminus X'| = 1$ 
	or $|X' \setminus X| = 1$. In the first case, $X$ is called an \emph{introduce bag}
	and, in the second, a \emph{forget bag}.
\item If $x$ is a node with two children $x_1$ and $x_2$ and if $X,X_1,X_2$ 	
	are the bags assigned to them, then $X = X_1 = X_2$. We call such a bag $X$ a
       \emph{join bag}.
\end{itemize}

\begin{proposition}[\cite{Bod93}]
  Given a graph $G$ with $n$ vertices and a tree decomposition of $G$ of
  width $w$ it is possible to compute a nice tree decomposition of $G$
  of width $w$ with at most $n$ bags in linear time.
\end{proposition}

The main property of tree decompositions that we will exploit is the
fact that each bag $X$ associated with an internal node 
is a vertex separator of $G$. Hence with each bag $X$ of a nice 
tree decomposition we can associate two (forget, introduce) or three 
(join) well-defined terminal subgraphs with terminal set $X$. 
For further information on treewidth and tree decompositions, we 
refer the reader to Bodlaender's survey~\cite{Bod93}.

A \emph{treedepth decomposition} of a graph $G$ is a pair $(F, \psi)$, where $F$ 
is a rooted forest and $\psi \colon V(G) \to V(F)$ in an injective mapping 
such that if $uv \in E(G)$ then either $\psi(u)$ is an ancestor of $\psi(v)$
or vice versa. Whenever we deal with treedepth decompositions in this paper, 
the mapping~$\psi$ will usually be implicit as we will have $V(G) \subseteq V(F)$.  

\begin{definition}[Treedepth]
  The \emph{treedepth} $\td(G)$ of a graph~$G$ is the minimum height of
  any treedepth decomposition of $G$. 
\end{definition}

We list some well-known facts about graphs of bounded treedepth.
Proofs that are omitted can be found in~\cite{NOdM12}.
\begin{enumerate}
\item If~$\td(G) \leq d$, then~$G$ has no paths with~$2^d$ vertices
  and, in particular, any DFS-tree of~$G$ has depth at most~$2^d - 1$.
\item If~$\td(G) \leq d$, then $\tw(G)\leq \pw(G) \leq d-1$.  
\end{enumerate} 
A useful way of thinking about graphs of bounded treedepth is that
they are (sparse) graphs with no long paths.

A treedepth decomposition of a graph is not unique. One can always add
extra vertices to a treedepth decomposition without increasing its
height. We introduce the notion of \emph{trivially improvable
  treedepth decomposition} so that we can differentiate between
treedepth decomposition which have such unnecessary nodes and those
who do not.

\begin{definition}[Trivially Improvable Treedepth Decompositions\label{def:trivial-tdd}]
  A treedepth decomposition $T$ of a graph $G$ is \emph{trivially
  improvable} if $V(G) \subsetneq V(T)$.
\end{definition}

We will also use extensively a special kind of treedepth
decompositions that we will call \emph{nice treedepth
  decompositions}. This notion is similar to that of \emph{minimal
  trees} introduced in \cite{FGP13}.

\begin{definition}[Nice Treedepth Decomposition\label{def:nice-tdd}]
  A treedepth decomposition~$T$ of~$G$ is \emph{nice} if the following
  conditions are met: 
  
  \begin{itemize}
  \item $T$ is not trivially improvable.
  \item For every node $x \in V(T)$, the subgraph of~$G$ induced by
    the nodes in $T_x$ is connected.
  \end{itemize}
\end{definition}

\subsection{Fixed Parameter Tractability}
\label{sec:fpt}

Parameterized complexity deals with algorithms for decision problems whose
instances consist of a pair $(x,k)$, where~$k$ is a secondary measurement known
as the \emph{parameter}. A major goal in parameterized complexity is to
investigate whether a problem with parameter~$k$ admits an algorithm with
running time $f(k) \cdot |x|^{O(1)}$, where~$f$ is a function depending only on
the parameter and $|x|$ represents the input size. Parameterized problems that
admit such algorithms are called \emph{fixed-parameter tractable} and the class
of all such problems is denoted \FPT. For an introduction to the area
see~\cite{DF99,FG06,Nie06}.

Even if we will not explicitly mention it in the rest of the paper, it is
clear that all the algorithms presented in this paper are in
\FPT\ (w.r.t.~the appropriate parameter). Algorithm~\ref{fig:main-algorithm-tree} 
is in \FPT parameterized by the treedepth $t$ and the width $w$ of the 
given tree decomposition and Algorithms~\ref{fig:simple-algorithm} 
and~\ref{fig:fast-algorithm} are in \FPT\ parameterized by the treedepth~$t$.


\section{Dynamic Programming Algorithm}\label{sec:Algorithm}
In this section we present an algorithm which takes as input 
a triple $(G, \treedecomp,t)$, where $G$ is a graph $G$, 
$\treedecomp$ a tree decomposition of $G$ of width~$w$, and $t$ 
an integer, and decides whether $\td(G) \leq t$ 
in time $2^{O(wt)} \cdot n$. For yes-instances, the algorithm 
can be modified to output a treedepth decomposition by backtracking. 
Later we will show how this algorithm can easily be
used to achieve the three claimed results.

\subsection{Main Algorithm}
\label{sec:alg-on-tw}


Our algorithm is a dynamic programming algorithm. It works by creating
tables of \emph{partial decompositions}. Every operation of the
algorithm will take one or two sets of partial decompositions and
create a new set of partial decompositions. More specifically, such an
operation will be done for every bag of the tree decomposition. These
partial decompositions will represent treedepth decompositions and
they will have the same height as the treedepth decompositions they
represent.

As such, we begin by defining partial decompositions.

\begin{definition}[Partial decomposition]
  A \emph{partial decomposition} is a triple $(F,X,h)$, where
	\begin{itemize}
        \item $F$ is a forest of rooted trees with $X \subseteq V(F)$; and,
        \item $h \colon V(F) \to \mathbb{N}^+$ is a \emph{height function}
		which obeys the property that for nodes $x,y \in V(F)$
		with~$x$ an ancestor of~$y$, $h(x) > h(y)$.
  \end{itemize}
\end{definition}

\noindent Since we are going to use partial decompositions to represent
treedepth decompositions of a graph we need to introduce a notion of
height:

\begin{definition}[Height of a partial decomposition]\label{def:height-partial-decomposition}
  Let $(F,X,h)$ be a partial decomposition and let $R$ be the
  set of all roots in $F$. The height of $(F,X,h)$ is
  $\max_{x \in R} h(x)$.
\end{definition}

The following definition will be the key to keep the sets during the
dynamic programming small.

\begin{definition}[Partial decomposition equivalency]\label{def:equivalency}
  Two partial decompositions
  $(F_1,X_1,h_1)$ and $(F_2,X_2,h_2)$ are \emph{equivalent} if and only if 
  $X_1 = X_2$ and there exists a bijective function $\psi \colon V(F_1) \to V(F_2)$ such that 
  the following hold:
  \begin{itemize}
  \item the function $\psi$ is an isomorphism between $F_1$ and $F_2$;
  \item for all $x \in X_1$, $\psi(x) = x$, that is, $\psi$ is the identity map 
	when restricted to the set $X_1$;
  \item for every node $v$ in the forest $F_1$, $h_1(v) = h_2(\psi(v))$.
  \end{itemize}
\end{definition}

Clearly two equivalent partial decompositions have the same
height. Later we will show that it suffices to keep a representative
for each equivalency class during the dynamic programming. We will do
this by proving that we only care about a specific restricted type of
treedepth decompositions and showing that we have a partial
decomposition representing everyone of these important treedepth
decompositions. Since their height are the same, we can read the
heights of all important treedepth decompositions from the heights of
the partial decompositions representing them.

It should be noted that the algorithm is oblivious to the fact that
there is a connection between partial decompositions and treedepth
decompositions, \ie~the properties that connect partial decompositions
with treedepth decompositions are only implicit. This works because we
can disregard every part a treedepth decompositions that does not
contain a node of the bag we are currently working on.

The following definition will be helpful to decrease the size of the
tables when the size of the bags decreases. The way in which we will
connect treedepth decompositions to partial decompositions will be
based on this operation.

\begin{definition}[Restriction of a partial decomposition]\label{def:restriction-partial}
  The \emph{restriction of a partial decomposition $(F,X,h)$ to $X' \subseteq X$} is 
  the partial decomposition $(F',X',h')$, where $F'$ is obtained by iteratively deleting the leaves 
  of the forest $F$ that are \emph{not} in $X'$. The height function
  $h'$ is obtained from $h$ by restricting it to $V(F')$.
\end{definition}

As we move from the leaves to the root of the tree decomposition we
will need a relationship between the previous set of partial
decompositions and the new ones we want to compute. The most important
thing that we have to make sure of is that the predecessor
relationship is maintained, since this is what permits there to be an
edge in a treedepth decomposition. The following definitions will be
used to make sure that the predecessor relation is kept intact.

\begin{definition}[Topological generalization]
  Let $F_1,F_2$ be rooted forests and let $X$ be a set of vertices such that
  $X \subseteq V(F_1) \cap V(F_2)$. We say $F_1$ \emph{topologically
    generalizes} $F_2$ under $X$ if there exists an injective mapping
  $f \colon V(F_2) \to V(F_1)$ where the following conditions hold:
  \begin{itemize}
  \item $f |_{X} = \mathop{id}$.
  \item For any node $x \in V(F_2)$ and an ancestor $y$ of $x$, 
	$f(y)$ is an ancestor of $f(x)$ in $F_1$
  \end{itemize}

  \noindent We say that a partial decomposition $(F_1,X_1,h_1)$ topologically generalizes
  a partial decomposition $(F_2,X_2,h_2)$ if $X_2 \subseteq X_1$ and 
  $F_1$ topologically generalizes  $F_2$ under $X_2$.
\end{definition}

We will also show that it suffices to work on rooted graphs. This is
not fundamental to the algorithm, but it will make its description and
proof of correctness much easier.

\begin{definition}[Rooted graph]
  A \emph{rooted graph} $G = (V,E,r)$ is a graph with the specified
  \emph{universal vertex} $r \in V(G)$ which is connected to every
  other vertex of $G$.
\end{definition}

These definitions are all that are needed to describe the algorithm. We
next define the operations we will perform on the bags of a tree
decomposition of a rooted graph $G$. We try to describe these operations in a succinct
manner. It is not immediately obvious why these operations should
provide the previously described full characterization of all relevant
treedepth decompositions. This will be fully explained in the next
subsection.

During the forget operation we restrict every partial decomposition
that we have on a smaller set and only keep one representative for
every equivalency class. This will make our tables smaller.

\begin{definition}[Forgetting a vertex from a partial decomposition]\label{def:node-deletion}
  Let $G$ be a graph, let $X \subseteq V(G)$ and let $R'$ be a 
  set of partial decompositions on the set $X$. For a vertex 
  $u \in X$, the forget operation on $u$ denoted by $\forget(R',X,u)$ 
  is defined to be a set $A$ of partial decompositions obtained as follows:  
  initially set $A \leftarrow \varnothing$; for every partial decomposition 
  $(F',X',h') \in R'$, consider its restriction to the set $X \setminus \{u\}$ and 
  add it to the set $A$ only if it is \emph{not equivalent} to any member in $A$. 
\end{definition}

Note that the set $A$ is not unique and that it contains only
non-equivalent partial decompositions obtained by restricting the
decompositions in $R'$ to $X \setminus \{u\}$.

The introduce operation is somewhat more involved. Given a set $R'$ 
of partial decompositions of the form $(F', X', h')$ where $X' \subseteq V(G)$, 
the result of introducing $u \in V(G) \setminus X'$ is a set $A$ of partial decompositions 
whose elements $(F, X, h)$ are computed as follows:

\begin{enumerate}
\item Guess every forest $F$ which complies with certain conditions.
\item Find a partial decomposition $(F',X',h') \in R'$ such that
  $F$ topologically generalizes $F'$. If no such partial decomposition
  exists, discard $F$.
\item Given $F$ and $F'$, for every function $f$ that witnesses 
  $F$ topologically generalizing $F'$, create a partial decomposition 
  of the form $(F,X = X' \cup \{u\},h)$.
\item Add $(F,X,h)$ to $A$ if its height is smaller than $t$ and there
  is no equivalent partial decomposition already contained in $A$.
\end{enumerate}

Formally this translates to the following:

\begin{definition}[Vertex introduction into a partial decomposition]\label{def:introduction}
  Let $G = (V,E,r)$ be a rooted graph, let $X' \subseteq V(G)$ and let $R'$ be a
  set of partial decompositions of the form $(F',X',h')$. For a vertex 
  $u \in V(G) \setminus X'$ and an integer~$t$, the introduction operation on $u$
  denoted by $\intro_t(R', X', u,G)$ is defined to be a set~$A$ of partial 
  decompositions constructed as follows:

  Let $X = X' \cup \{u\}$. Initialize $S \leftarrow
  \varnothing$. Generate every tree $F$ with up to $|X| \cdot t$
  vertices which fulfills the following properties:
  \begin{itemize}
  \item $r$ is the root of $F$.
  \item $X \subseteq V(F)$.
  \item All leafs of $F$ are in $X$.
  \item $E(G[X]) \subseteq E(clos(F)[X])$. 
  \end{itemize}

  For every partial decomposition $(F',X',h') \in R'$ and every 
  function $f \colon V(F') \to V(F)$ that witnesses that $F$ 
  topologically generalizes $F'$ on the set $X \setminus \{u\}$ 
  add $(F,(F',X',h'),f)$ to $S$ if $f(F') = V(F)\setminus \{u\}$.

  For every $(F,(F',X',h'),f) \in S$ compute the partial decomposition
  $(F,X,h)$, where $h$ is defined recursively by visiting the
  vertices of $F$ in depth-first post-order fashion. Let $z \in F$ and let $C$
  be the set of children of $z$ in $F$. When $z$ is visited, if $z
  \neq u$ and there exists a node $z' \in V(F')$ such that $f(z') = z$,
  set $h(z) = \max \{\max_{c \in C} h(c) + 1, h'(z')\}$. Else for any
  other node $z \in V(F)$ set $h(z) = \max_{c \in C} h(c) +
  1$.\footnote{We define the maximum over the empty set to be zero.}
  Finally add the partial decomposition $(F,X,h)$ to the set $A$, if
  its height is smaller that $t$ and $A$ does not contain an
  equivalent partial decomposition to $(F,X,h)$.
\end{definition}

\begin{definition}[Joining Partial Decompositions]
  Let $G = (V,E,r)$ be a rooted graph. Let $R_1$ and $R_2$ be two sets
  of partial decompositions on $X \subseteq V(G)$.  Let
  $t$ be an integer. Then the join operation $\join_t$ is defined via
  $\join_t(X, R_1, R_2,G) = A$, where $A$ is a set of partial
  decompositions which is constructed as follows:

  Initialize $S \leftarrow \varnothing$. Generate every tree $F$ with up to
  $|X| \cdot t$ vertices which fulfills the following properties:
  \begin{itemize}
  \item $r$ is the root of $F$.
  \item $X \subseteq V(F)$.
  \item All leaves of $F$ are in $X$.
  \end{itemize}
  
  Take every pair of partial decompositions $(F_1,X,h_1) \in R_1$ and
  $(F_2,X,h_2) \in R_2$ and every pair of functions $f_1$ and $f_2$
  which witness that $F$ topologically generalizes $F_1$ and
  $F_2$ on the set $X$ respectively. Add the tuple
  $(F,(F_1,X,h_1),(F_2,X,h_2),f_1,f_2)$ to $S$ if $f_1(F_1) \cap f_2(F_2)= X$
  and $f_1(F_1) \cup f_2(F_2) = V(F)$.

  For every $(F,(F_1,X,h_1),(F_2,X,h_2),f_1,f_2) \in S$ we get one
  partial decomposition $(F,X,h)$ where $h$ is defined as follows: The
  function $h$ is defined recursively by visiting the vertices
  of $F$ in depth-first post-order fashion. Let $z \in F$ and let $C$ be the
  set of children of $z$ in $F$. Let $\alpha_1 = h_1(z_1)$ if there exists 
  $z_1$ such that $f_1(z_1) = z$ and $\alpha_1 = 1$ otherwise. Analogously,
  let $\alpha_2 = h_2(z_2)$ if there exists $z_2$ such that $f_2(z_2) = z$ and $\alpha_2 = 1$ otherwise. Then
  we compute the height of $z$ as $h(z) = \max \{ \max_{c\in C} h(c)+1,\alpha_1,\alpha_2 \}$.

  Finally add the partial decomposition $(F,X,h)$ to the set $A$, if
  its height is smaller that $t$ and $A$ does not contain an
  equivalent partial decomposition to $(F,X,h)$.
\end{definition}

\begin{figure}
\begin{minipage}{\textwidth}
\begin{algorithm}[H]
  
  \small
  \caption{treedepth{\label{fig:main-algorithm-tree}\sc }}
  \KwIn{A graph $G'$, an integer $t$ and a nice rooted tree decomposition $\mathcal T'$ of $G'$
        with root bag~$X$.}
  \KwOut{\texttt{True} if the treedepth of $G'$ is at most $t$ and \texttt{False} otherwise.}  \BlankLine

  Add a universal vertex $r \notin V(G')$ to the graph $G'$ to obtain $G$\;
  Obtain a nice tree decomposition $\mathcal T$ of $G$ as follows:\;
  ~~~~~~~start with $\mathcal T = \mathcal T'$\;
  ~~~~~~~add $r$ to every bag of $\mathcal T$\; 
  ~~~~~~~for every leaf bag of $\mathcal T$, add $\{r\}$ as a child-bag\;
  
  $R =$ treedepth-rec($G, t+1, \mathcal T, X)$\label{alg:top-rec-call}\;
  \Return $R \neq \emptyset$\label{alg:final-return}\;
\end{algorithm}\vspace{25pt}

\begin{algorithm}[H]
  \small
  \caption{treedepth-rec{\label{fig:main-algorithm-tree-rec}\sc }}
  \KwIn{A rooted graph $G = (V,E,r)$, an integer $t$ and a tree
    decomposition $\mathcal T$ of $G$ containing $r$ in every bag and
    a bag $X$ of $\mathcal T$.}

  \KwOut{A set $R$ of partial decompositions.}

  \BlankLine

  $R = \emptyset$\;
  \BlankLine

  \If{$X$ is a leaf\label{alg:X-is-leaf}}
  {
    $r =$ the only vertex contained in $X$\;
    $F =$ a tree consisting of just the node $r$\;
    $h$ is a function which is only defined for $r$ and $h(r) = 1$\;
    $R = \{(F,\{r\},h)\}$\label{alg:first-step}\;
  }\ElseIf{$X$ is a forget bag\label{alg:X-is-forget}}
  {
    $u =$ forgotten vertex\;
    $X' =$ the child of $X$\;
    $R' =$ treedepth-rec($G,t,\mathcal T,X'$)\;
    $R = \forget(R',X',u)$\;
  }\ElseIf{$X$ is an introduce bag\label{alg:X-is-introduce}}
  {
    $u =$ introduced vertex\;
    $X' =$ the child of $X$\;
    $R' =$ treedepth-rec($G,t,\mathcal T,X'$)\;
    $R = \intro_{t}(R',X',u,G)$\;
  }\ElseIf{$X$ is a join bag\label{alg:X-is-join}}
  {
    $\{X_1, X_2\} =$ the set of children of $X$\;
    $R_1 =$ treedepth-rec($G,t,\mathcal T,X_1$)\;
    $R_2 =$ treedepth-rec($G,t,\mathcal T,X_2$)\;
    $R = \join_{t}(X,R_1,R_2,G)$\;
  }
  \Return $R$\label{alg:return-rec}\;
\end{algorithm}
\end{minipage}
\end{figure}

The main algorithm can be found in
Algorithm~\ref{fig:main-algorithm-tree}. We claim that this algorithm 
correctly decides, given an $n$-vertex graph $G$ and a tree decomposition 
of width at most $w$, whether $G$ has treedepth at most~$t$ in time $2^{O(wt)} \cdot n$. 
The rest of this section is dedicated to proving this statement.

\subsection{Correctness of Algorithm~\ref{fig:main-algorithm-tree}}
\label{sec:correctnes-alg-tw}

To prove the correctness of this algorithm we will show a relationship
between treedepth decompositions and restrictions. Our proof can be
divided in the following steps:

\begin{enumerate}
\item We show that every graph $G$ admits a nice treedepth 
	decomposition of height $\td(G)$ (Lemma~\ref{lemma:nice-td-exists}).
\item We show that it is sufficient to work with rooted graphs and that 
	such graphs have an optimal nice treedepth decomposition $T$ such that 
	root of $T$ is the root of graph (Lemma~\ref{lemma:root-always-root}).
\item We define the \emph{restriction of a tree}. Since in this context we treat
  tree decompositions as trees, this will provide a relationship
  between treedepth decompositions and partial decompositions (Definition~\ref{def:restriction-treedepth}).
\item We use the previous facts to show that for any nice treedepth decomposition
      of the graph, our table contains its restriction (Lemma~\ref{lemma:all-restrictions-contained});
\item and moreover every partial decomposition contained in the table is a 
      restriction of some treedepth decomposition of the graph (Lemma~\ref{lemma:all-valid-restrictions}).
\end{enumerate}

All this together achieves the desired result. 

We will prefer to work with decompositions that are \emph{not}
trivially improvable. The next lemma shows that one can always obtain
such a decomposition from a trivially improvable one without
increasing the height.

\begin{lemma}
  \label{def:node-deletion-trivial}
  Let $T$ be a trivially improvable treedepth decomposition of a graph
  $G$ of height~$h$. Let $x \in V(T) \setminus V(G)$ be a root of some
  tree in the decomposition $T$. Then the decomposition obtained by
  removing~$x$ is a treedepth decomposition of~$G$ with height at
  most~$h$.
\end{lemma}
\begin{proof}
  Since~$x \notin V(G)$, we have that $G \subseteq \clos(T \setminus
  \{x\})$. Thus $T \setminus \{x\}$ is a treedepth decomposition of
  $G$. Clearly the height does not increase on deleting~$x$.
\end{proof}

\begin{lemma}
  \label{def:edge-contraction-trivial}
  Let~$T$ be a trivially improvable treedepth decomposition of a
  graph~$G$ with height~$h$. Suppose that $x \in V(T) \setminus V(G)$ 
  be a non-root node and let $y$ be its parent in $T$.
  Then the treedepth decomposition obtained by contracting the edge $xy$
  is a treedepth decomposition of $G$ with height at most~$h$.
\end{lemma}
\begin{proof}
  Suppose $T'$ is the forest obtained by contracting the edge $xy$. Then the 
  height of $T'$ is at most~$h$. If $a,b \in V(T)$ is an ancestor-descendant 
  pair that represents an edge of $G$, then these vertices 
  form an ancestor-descendant pair in $T'$ too. Thus $T'$ is a 
  treedepth decomposition of $G$ with height at most~$h$. 
  
\end{proof}

\begin{corollary}
  \label{def:edge-contraction-trivial-operation}
  Given a trivially improvable treedepth decomposition~$T$
  of a graph $G$, one can obtain a decomposition of~$G$ that is not trivially 
  improvable in time polynomial in $|T|$. 
\end{corollary}
\begin{proof}
  Apply either Lemma~\ref{def:node-deletion-trivial} or~\ref{def:edge-contraction-trivial} 
  until $V(T) = V(G)$.
\end{proof}

The operations described in Lemma~\ref{def:node-deletion-trivial} and
Lemma~\ref{def:edge-contraction-trivial} does not increase the height
of a decomposition. It therefore suffices to work with decompositions
that are not trivially improvable. We will now use these results to
proof certain properties of nice treedepth decomposition. In a sense,
nice treedepth decompositions are those whose \emph{structure} cannot
be easily improved.

\begin{lemma}\label{lemma:nice-td-exists}
  Every graph $G$ admits a nice treedepth decomposition of height $\td(G)$.
\end{lemma}
\begin{proof}
  Let us assume~$G$ to be connected. If~$G$ has more than one connected 
  component then we can apply this argument to each component in turn. 
  
  By Corollary~\ref{def:edge-contraction-trivial-operation}, it is sufficient 
  to show that given an optimal treedepth decomposition that is not trivially 
  improvable, one can construct a decomposition of the same height that is
  nice. Therefore let $T$ be an optimal decomposition of $G$ with root~$r$ 
  that is not trivially improvable and let $x \in V(T)$ be a node at which the niceness 
  condition is violated. That is, the subgraph $G[V(T_x)]$ of~$G$ induced by the 
  vertices in the subtree of~$T$ rooted at~$x$ has as connected components $C_1, \ldots, C_l$, 
  where $l \geq 2$, and suppose that $x \in V(C_1)$. 
   Note that~$T_x$ itself is a trivially improvable decomposition of $C_i$ for each 
  $1 \leq i \leq l$. Repeatedly use Lemma~\ref{def:edge-contraction-trivial-operation} 
  to obtain a decomposition~$T_i$ for~$C_i$ which is not trivially improvable such that 
  $\height{}{T_i} \leq \height{}{T_x}$. 
  Note that the root of~$T_1$ is~$x$.

  If $x'$ is the parent of $x$ in~$T$, then for every component $C_i$, there exists at 
  least one node ~$y$ in the $(r,x')$-path in~$T$ such that $y$ has an edge to $C_i$ 
  in the graph $G$. This follows because $G$ is connected and the only vertices that 
  $C_i$ can be connected to are on the $(r,x')$-path in~$T$. Thus for each component 
  $C_i$, we can identify a node $y_i$ on the $(r,x')$-path such that $y_i$ has the 
  maximum distance from~$r$ among all vertices on the path that are connected to $C_i$
  in the graph $G$. Construct a new tree~$T'$ from $T$ by deleting $T_x$ and, for $1 \leq i \leq l$, 
  attaching $\rootof{}{T_i}$ to $y_i$. We claim that $T'$ is a treedepth decomposition of~$G$ of height 
  at most $\height{}{T}$; that the subgraph of~$G$ induced by the vertices in the 
  subtree $T'_x$ is connected; and, that if this construction is repeated on the tree $T'$
  to obtain a tree $T''$, then $T''_x$ induces a connected subgraph of~$G$. 

  That $T'$ is a treedepth decomposition of~$G$ is easy to see as we connected $T_i$ 
  to the ``deepest node'' on the $(r,x')$-path that has an edge to~$C_i$. Consequently, all 
  neighbors of $C_i$ on the $(r,x')$-path are ancestors of the root of~$T_i$ in~$T'$. 
  We therefore have $G \subseteq \clos(T')$. The height of $T'$ cannot increase since   
  each $T_i$ has height at most $\height{}{T_x}$ and they are connected to vertices on 
  the $(r,x')$-path (which are ``above'' the node~$x$). What is also clear is that 
  $T'_x = T_1$ induces a connected subgraph of~$G$. Suppose that this procedure is repeated 
  on a node $z \in V(T')$ to obtain $T''$. 

  We distinguish two cases. First suppose that $z \notin V(T'_x)$. To construct $T''$, one 
  would delete $T'_z$ from $T'$ and add trees $T'_1, \ldots, T'_p$
  to vertices on the $(r,z')$-path in $T'$, where $z'$ is the parent of $z$ in~$T'$. The crucial 
  observation here is that since $z \notin V(T'_x)$, $x$ does not appear in the $(r,z')$-path.
  For if $x$ did appear on this path, we would have had $z \in V(T'_x)$, contradicting our 
  assumption that this is not the case. Note that it might be that $T'_x$ is a subtree of $T'_z$ 
  and therefore one of the connected components, say $C'_1$, of $G[V(T'_z)]$ contains the 
  vertices of $T'_x$ (and perhaps more). Construct a treedepth 
  decomposition of $C'_1$ whic is not trivially improvable by starting out with $T'_z$ and using Lemma~\ref{def:edge-contraction-trivial-operation}
  to remove redundant vertices. Call this decomposition $T'_1$ and observe that $T'_x$ is a 
  subtree of it. Construct decompositions $T'_2, \ldots, T'_p$ for the remaining components $C'_2, \ldots, C'_p$ 
  of $G[T'_z]$.  Then, as before, add trees $T'_i$ to the deepest node on the $(r,z')$-path that has an edge to $C'_i$.
  Since this did not modify $T'_x$, we have $T''_x = T'_x$. 

  Next suppose that $z \in V(T'_x)$. Then $z \neq x$, since the subgraph of $G$ induced by $T'_x$ 
  is connected. Thus $z$ must lie ``deeper'' in the tree $T'_x$. Let $z'$ be the parent of $z$ in $T'$
  ($z'$ may be $x$). Suppose that $T'$ is modified by deleting $T'_z$ and adding the trees $T'_1, \ldots, T'_p$
  to vertices $z_1, \ldots, z_p$ on the $(r,z')$-path. Recall that that $z_i$ is the ``deepest node'' 
  on the $(r,z')$-path to which $T'_i$ has an edge (when viewed as vertices of $G$). We claim that 
  each $z_i$ is a descendant of $x$. Had this not been the case then some $z_i$ would be an ancestor 
  of $x$ in $T'$, and the subgraph of~$G$ induced by $T'_x$ would not have been connected. What this shows 
  is that once a node has been handled, it no longer has to be handled again. Therefore by repeating 
  this procedure at most $|V(G)|$ many times, we can obtain a treedepth decomposition which is nice
  and of height at most $\height{}{T}$. 
  Since the time taken to effect this transformation per node is polynomial in $|G|$, the overall 
  time taken to transform a decomposition which is not trivially improvable to a nice decomposition is also 
  polynomial in $|G|$.
  \end{proof} 

As a result of Corollary~\ref{def:edge-contraction-trivial-operation} and the proof
of Lemma~\ref{lemma:nice-td-exists}, we obtain the following easy-to-prove, 
yet, important result. 

\begin{corollary}
  \label{corollary:nti-and-nice-poly}
  Let~$T$ be a treedepth decomposition of a graph~$G$. One can compute
  in time polynomial in $|G|$, a nice treedepth decomposition $T'$
  with the following properties:
  \begin{enumerate}
  \item $\height{}{T'} \leq \height{}{T}$;
  \item for each vertex $x \in V(G)$, $\height{T'}{x} \leq
  \height{T}{x}$;
  \item for any node $x \in V(T')$, we have that $\anc{T'}{x} \subseteq \anc{T}{x}$ 
	and $\des{T'}{x} \subseteq \des{T}{x}$.
  \end{enumerate}
\end{corollary}

Given that one can transform any treedepth decomposition~$T$ into one 
that is nice and not trivially improvable in time polynomial in $|V(T)|$, 
we will henceforth assume that the treedepth decompositions that we deal 
with have this property. 
Here is another property of treedepth decompositions that
will prove to be useful to us later.

\begin{lemma}
  \label{lemma:connection-to-child-subtree}
  Let $T$ be a nice treedepth decomposition of a graph $G$. Let $x \in
  V(G)$ be a vertex such that $x$ is not a leaf in~$T$. If $y$ is a
  child of $x$ in $T$, then there exists an edge $x c \in E(G)$, for
  some $c \in V(T_y)$.
\end{lemma}
\begin{proof}
  Since $T$ is a nice treedepth decomposition, the subtree
  $T_x$ rooted at $x$ induces a connected subgraph of~$G$. From
  the definition of a treedepth decomposition, it follows that there can
  be no edge in $G$ adjacent to a node of $V(T_y)$ and a node of $(V(T_x)
  \setminus \{x\}) \setminus V(T_y)$. From this it follows that for
  $G[V(T_x)]$ to be connected, there must be an edge
  between $x$ and some node of $T_y$.
\end{proof}
Thus every inner node in a nice treedepth decomposition has an edge to
at least one of its descendants (in the graph represented by the
decomposition).

\begin{lemma}
  \label{lemma:set-children-single-component}
  Given a nice treedepth decomposition~$T$ of a graph $G$, let $x \in
  V(G)$ and let $C = \child{T}{x}$.  For $C' \subseteq C$, let
  $T_x^{C'}$ denote the tree obtained from $T_x$ by deleting the
  subtrees rooted at the vertices of $C \setminus C'$.  Then
  $G[V(T_x^{C'})]$ is a connected subgraph of~$G$.
\end{lemma}
\begin{proof}
  Since $T$ is a nice treedepth decomposition, it follows that for
  every $c \in C$ the subtree~$T_c$ of $T$ rooted at $c$ induces a
  connected subgraph of~$G$. From Lemma~\ref{lemma:connection-to-child-subtree},
  it follows that~$x$ is connected to a node of~$T_c$. Thus the lemma follows.
\end{proof}

As mentioned before, our algorithm works on partial decompositions. 
Let us first define a relation between treedepth decompositions and
partial decompositions by defining what we call the \emph{restriction}
of a tree. Remember that we treat treedepth decompositions as trees,
and we will do so often.

\begin{definition}[Restriction of a tree \label{def:restriction-treedepth}]
  Given a tree~$T$, let~$(T, V(T), h)$ be the partial 
  decomposition where $h(x) = \height{T}{x}$ for all $x \in V(T)$. For $X \subseteq V(T)$, 
  let $(F,X,h)$ be the restriction of $(T,V(T),h)$ to the set~$X$.
  A partial decomposition $(F',X,h')$ is a \emph{restriction} of $T$ if $(F',X,h')$ is
  equivalent to $(F,X,h)$. We call the function $\psi \colon V(F') \to V(F)$ that
  witnesses the equivalency as per Definition~\ref{def:equivalency} of these
  two restrictions the \emph{witness of the restriction}. 
\end{definition}

The following properties of restrictions will prove to be useful later on.

\begin{lemma}\label{lemma:restriction-transitive}
  Let $(F,X,h)$ be a partial decomposition.
  For $X' \subseteq X$, let $(F', X', h')$ be the restriction of $(F,X,h)$ to $X'$. 
  Then for any $X'' \subseteq X'$, the restrictions of $(F',X',h')$ and
  $(F,X,h)$ to $X''$ are identical.
\end{lemma}
\begin{proof}
  First observe that if~$x$ is a leaf in~$F$ then for any $y \neq x$, 
  $x$ is a leaf in $F - y$. Moreover if we restrict the decomposition $(F,X,h)$ 
  to $X''$, then the only leaves of the forest are elements of~$X''$. 
  Suppose that the restrictions of $(F',X',h')$ and $(F,X,h)$ 
  to~$X''$ yields (respectively) the decompositions 
  $(\widetilde{F}',X'',\widetilde{h}')$ and $(\widetilde{F},X'',\widetilde{h})$. 
  Let $s_1 = v_1, \ldots, v_p$ be the sequence in which vertices were 
  deleted to obtain $(\widetilde{F}',X'',\widetilde{h}')$ from $(F,X,h)$;
  and, $s_2 = w_1, \ldots, w_q$ were the vertices that were deleted to obtain
  $(\widetilde{F},X'',\widetilde{h})$ from $(F,X,h)$.

  Suppose that there exists a node $y$ in the sequence $s_1$ that does \emph{not}
  occur in~$s_2$ and suppose that $v$ is the first such node of $s_1$ so 
  that $s_1 = v_1, \ldots, v_l, v, \ldots$. Note that $v$ is a leaf after 
  the vertices $v_1, \ldots, v_l$ are deleted, irrespective of the order of deletion. 
  Since the vertices $v_1, \ldots, v_l$ occur in $s_2$, suppose that $w_i$ is the last 
  of these that occur in $s_2$. Then after the deletion of $w_i$ (in the sequence $s_2$), 
  the node $v$ continues to remain as a leaf and this fact does not change with
  further deletions down the sequence~$s_2$. But $v$ was deleted in the sequence $s_1$
  and hence $v \notin X''$ and the fact that $v$ does not appear in the sequence $s_2$
  implies that $\widetilde{F}$ has a leaf node that is not an element of~$X''$, a contradiction.
  This shows that every node of~$s_1$ appears in~$s_2$. Reversing the argument, one 
  sees that every node in~$s_2$ appears in~$s_1$. Hence $s_1$ and $s_2$ contain the same
  vertices, possibly in a different order. Therefore $V(\widetilde{F}) = V(\widetilde{F}')$
  and the partial decompositions $(\widetilde{F}',X'',\widetilde{h}')$ and 
  $(\widetilde{F},X'',\widetilde{h})$ are identical.    
\end{proof}

Lemma~\ref{lemma:restriction-transitive} immediately implies the following:

\begin{corollary}\label{corollary:restriction-unique}
  Let $(F,X,h)$ be a partial decomposition and
  let $X' \subseteq X$. The restriction of $(F,X,h)$ on $X'$ is
  unique up to isomorphism.
\end{corollary}

From Corollary~\ref{corollary:restriction-unique}, it follows
immediately that the restriction of a treedepth decomposition to a set
is unique up to isomorphism.  Importantly, the number of vertices
in such a forest is at most $|X| \cdot t$, where $t$ is the treedepth 
of the graph. This follows since every leaf of the forest is 
an element of $X$ and the number of vertices from any root to leaf path is at most~$t$. 

\begin{lemma}\label{lemma:equal-height-depth-partial}
  Let $(F,X,h)$ be a partial decomposition and let $(F',X',h')$ be the 
  restriction of $(F,X,h)$ to $X' \subseteq X$.
  Then for all $v \in F'$ it holds that $h'(v) = h(v)$. 
  Furthermore, $\height{}{F',X',h'} = \height{}{F,X,h}$.
\end{lemma}
\begin{proof}
  By the definition of a restriction: the height function of 
  $(F',X',h')$ is defined as the restriction of~$h$ to $V(F')$. 
  As such, for all $v \in V(F')$ we have that $h'(v) = h(v)$. 
  Also, $\height{}{F',X',h'} = \max_{x \in \rootof{}{F'}} h'(x)$,
  and since the sets of root nodes of the forests~$F$ and~$F'$ are the same,
  we have: $\height{}{F',X',h'} = \height{}{F,X,h}$. 
  \end{proof}

\begin{corollary}\label{corollary:restriction-structure}
  Let $(F,X,h)$ be a restriction of a treedepth decomposition $T$ of
  the graph $G$ such that $X \neq \emptyset$. Then the the height of
  $(F,X,h)$ is equal to the height of $T$.
\end{corollary}
\begin{proof}
  Follows immediately from Definition~\ref{def:restriction-treedepth}
  and Lemma~\ref{lemma:equal-height-depth-partial}.
\end{proof}

\begin{lemma}
  \label{lemma:restrictions-top-equiv}
  Let~$T$ be a (not trivially improvable) treedepth decomposition 
  of a graph~$G$, $X' \subseteq X \subseteq V(G)$, and, let 
  $F$ and $F'$ be the forests of the decomposition $T$ when 
  restricted to the sets $X$ and $X'$, respectively. Then $F$ 
  topologically generalizes~$F'$ under~$X'$. 
\end{lemma}
\begin{proof}
  Note that $V(F') \subseteq V(F)$ and hence the function $f \colon V(F') \to V(F)$ 
  defined by $f(x) = x$ for all $x \in V(F')$ witnesses that $F$ 
  topologically generalizes $F'$.
\end{proof}

As seen Algorithm~\ref{fig:main-algorithm-tree}, we will use the
contents of the bags of a tree decomposition as the set on which we
make restrictions on. If we have a restriction of a tree on the set $X$,
we would like to be able to bound the size of the restriction by $|X|$.
Since we never delete root nodes we cannot directly show that this is the case.
We will later see that by working on rooted we can overcome this problem.

\begin{lemma}\label{lemma:root-always-root}
  Let $G=(V,E,r)$ be a rooted graph with root~$r$. Then there is an optimal 
  treedepth decomposition~$T$ of~$G$ such that $\rootof{}{T} = r$.
\end{lemma}
\begin{proof}
  Suppose that $T'$ is an optimal treedepth decomposition of~$G$
  with root~$r' \neq r$ (since~$G$ is connected, $T'$ is actually a tree). 
  We assume that~$T'$ is not trivially improvable so that every node of~$T'$ 
  is a vertex of~$G$. Let $x_0, x_1, \ldots, x_p$ denote the vertices on 
  the $(r',r)$-path in~$T'$, where $x_0 = r'$ and~$x_p = r$. Then note that 
  since $T'$ is a treedepth decomposition and $r$ is a universal vertex, 
  for $0 \leq i \leq p-1$, $x_i$ has exactly one child~$x_{i+1}$ in~$T'$. 
  That is, $T'$ consists of the path $r', x_1, \ldots, x_{p-1}, r$ 
  with subtrees attached to $r$. Transform~$T'$ to obtain~$T$ by renaming the vertices
  of the $(r',r)$-path as follows: for $0 \leq i \leq p$, map 
  $x_i$ to $x_{(i+1) \mod (p+1)}$. In this transformation, vertex~$r$
  is made the ancestor of the vertices $r' = x_1, x_2, \ldots, x_{p}$ but 
  all other ancestor-descendant relationships are preserved and hence $T$ 
  is a treedepth decomposition of $G$. Moreover, $\height{}{T'} = \height{}{T}$,
  which is what we wished to prove. 
\end{proof}

\begin{lemma}\label{lemma:depth-rooted-graph}
  Let $G$ be a rooted graph obtained by adding
  a universal vertex~$r$ to a graph~$G'$. Then $\td(G) = \td(G') + 1$.
\end{lemma}
\begin{proof}
  To see that $\td(G) \leq \td(G') + 1$, take any optimal
  treedepth decomposition~$T'$ of~$G'$ and add edges between~$r$ 
  and the roots of the forest of~$T'$. This yields a treedepth 
  decomposition of $G$ of height $\td(G') + 1$. To see that $\td(G') \leq \td(G) - 1$,
  take an optimal treedepth decomposition~$T$ of~$G$ with $\rootof{}{T} = r$ 
  (Lemma~\ref{lemma:root-always-root} guarantees the existence of such a decomposition). 
  Now delete~$r$ from~$T$ to obtain a treedepth decomposition of~$G'$. 
\end{proof}

Lemma~\ref{lemma:root-always-root} motivates the following definition of
treedepth decompositions of rooted graphs.

\begin{definition}[Treedepth Decomposition of a Rooted Graph]
  A \emph{treedepth decomposition $T$ of a rooted graph} $G=(V,E,r)$ is a
  treedepth decomposition of $G$ such that $\rootof{}{T} = r$.
\end{definition}

Every entry in our table during the dynamic programming algorithm will be a
restriction of a treedepth decomposition. Let us define some
operations on partial decompositions, which we will later use to build
our tables during the dynamic programming.

\begin{lemma}\label{lemma:retained-nodes}
  Let $G = (V,E,r)$ be a rooted graph, let $G' = (V',E',r)$ be a
  rooted subgraph of $G$ and $X \subseteq V(G')$ be a set of
  nodes. Further, let $T$ be nice treedepth decomposition of $G$, and
  $T'$ be a nice treedepth decomposition of $G'$ computed from $T$ as
  in Corollary~\ref{corollary:nti-and-nice-poly}. Let $(F,X,h),
  (F',X,h')$ be respective restrictions of $T,T'$ to $X$. Then for
  every pair of functions $\psi$ and $\psi'$ that witness that
  $(F,X,h)$ is a restriction of $T$ to $X$ and $(F',X,h')$ is a
  restriction of $T'$ to $X$ respectively, it holds that $\psi'(F')
  \subseteq \psi(F)$.
\end{lemma}
\begin{proof}
  Assume to the contrary that there exist functions $\psi,\psi'$ such
  that there exists $v \in F'$ with $\psi'(v) \not \in \psi(F)$.
  First note that $v \not \in X$ as $X \subseteq \psi(F)$. Since
  $v$ is retained in $F'$, there exists a successor $y \in X$ of $v$ in $F'$.
  But by Corollary~\ref{corollary:nti-and-nice-poly}, the ancestor 
  relationship of vertices in $T'$ is preserved in $T$, therefore $y$ is also 
  a successor of $v$ in $T$. But then, by construction of $F$, the vertex
  $v$ must also be contained in $\psi(F)$.
\end{proof}

We are now going to prove certain properties of the algorithm, that
together will lead us to our main theorem.

\begin{lemma}\label{lemma:all-restrictions-contained}
  Let Algorithm~\ref{fig:main-algorithm-tree-rec} be called on
  $(G,t,\mathcal{T}, X)$, where $G$ is a graph rooted at~$r$, the
  remaining parameters $t, \mathcal{T}, X$ are as described in the
  algorithm. Then for every nice treedepth decomposition $T$ of height
  at most~$t$ (that is rooted at~$r$) of $G[V(\mathcal{T}_X)]$, the
  set $R$ returned by the algorithm contains a restriction of $T$ to
  the set~$X$.
\end{lemma}
\begin{proof}
  We will prove this by structural induction over tree decompositions:
  Consider the case that the tree decompositions consists of a single
  leaf bag. Remember that Algorithm~\ref{fig:main-algorithm-tree-rec}
  works on nice tree decompositions whose leafs contain a single
  vertex. The returned set $R$ then consists of the unique partial
  decomposition for a graph with a single vertex.

  \paragraph{Forget case}
  If $X$ is a forget bag whose single child in $\mathcal T$ is the bag
  $X'$, then the if-clause at line~\ref{alg:X-is-forget} is
  entered. By induction hypothesis, we assume that $R'$ contains a
  restriction to the set $X'$ of every nice treedepth decomposition
  $T$ rooted at $r$ of the graph $G[V(\mathcal T_{X'})]$. Fix such a
  $T$ and let $(F',X',h') \in R'$ be a restriction of $T$ to the set
  $X'$. Notice that $G[V(\mathcal T_{X'})] = G[V(\mathcal
  T_{X})]$. Therefore by Corollary~\ref{corollary:restriction-unique},
  we can restrict $(F',X',h')$ to the set $X$ to obtain the
  restriction $(F,X,h)$ of $T$ to $X$. By the definition of the forget
  operation (Definition~\ref{def:node-deletion}) the restriction of
  $(F',X',h')$ to $X$ is added to $R$.

  \paragraph{Introduce case}
  If $X$ is an introduce bag whose single child in $\mathcal T$ is the
  bag $X'$, then the if-clause at line~\ref{alg:X-is-introduce} is
  entered. Fix a nice treedepth decomposition $T$ rooted at $r$ of the
  graph $G[V(\mathcal T_X)]$ of height at most $t$. We want to show
  that a tuple $(F,X,h)$ is contained in $R$, which is a restriction
  of $T$ to $X$. Note that the treedepth decomposition $T$ is also a
  treedepth decomposition of $G[V(\mathcal T_{X'})]$. Let $T'$ be the
  nice treedepth decomposition of $G[V(\mathcal T_{X'})]$ computed
  from $T$ using Corollary~\ref{corollary:nti-and-nice-poly}. By
  induction hypothesis we assume that $R'$ contains the restriction
  $(F',X',h')$ of $T'$ to $X'$. Since $(F,X,h)$ is a restriction, the
  height of $F$ is at most $t$, its leaves are in $X$ and $r$ is its
  only root. Therefore $F$ has at most $|X| \cdot t$ vertices. This
  means that at some point the introduce operation will generate $F$,
  since all trees which comply with these characteristics are
  enumerated. By Lemma~\ref{lemma:restrictions-top-equiv} the tree $F$
  topologically generalizes $F'$. Thus a tuple $(F,X,h)$ will be added
  to the set $R$ of the introduce function from
  Definition~\ref{def:introduction}, and it is left to show that $h$
  is computed correctly.

  First let us prove that $u$ (the introduced vertex) either is an
  internal node of $F$ or its parent $y$ is contained in $X$. Let us
  assume that $u$ is not an internal node of $F$, \ie~it is a leaf
  of $F$, and that $y$ is not an element in $X$. By the properties of
  tree decompositions, the vertex $u$ can only have edges to vertices
  in $X$ in the graph $G[V(\mathcal T_X)]$. This means that $u$ must
  be a leaf in $T$: If $u$ had a subtree, it would not contain a
  vertex which is connected to $u$ in $G[V(\mathcal T_X)]$. This would
  by Lemma~\ref{lemma:connection-to-child-subtree} contradict our
  assumption that $T$ is a nice treedepth decomposition. But if we
  assume that the parent $y$ of $u$ is not an element of $X$, the
  subtree $T_y$ of $T$ rooted at $y$ would induce more than one
  connected component in $G[V(\mathcal T_X)]$, since $u$ cannot be
  connected to any node of $V(T_y) \setminus \{y\}$. Thus $u$
  has the stated property.

  We will now show that the height function $h$ is computed correctly
  for the leaves of $F$. Let $z$ be a leaf of $F$. As stated before if
  $z = u$ then $z$ must be a leaf in $T$. For this case we correctly
  set the height of $z = u$ to one, since it has no children in $F$
  either. For $z \neq u$ we know that the subtree $T_z$ of $T$ rooted
  at $z$ forms a single connected component in $G[V(\mathcal
  T_X)]$. Since $u \notin V(T_z)$ it follows that $T_z$ is still a
  single connected component in $G[V(\mathcal T_{X'})]$. Thus $T_z$ is
  also a subtree of $T'$. Accordingly we set the value correctly to
  $h(z) = \height{T'}{z} = h'(z)$.

  We have shown that we set the correct height for the leaves of $F$,
  let us now show inductively that the height is also set correctly
  for any internal node $z$ of $F$. By induction we can assume that
  the height is set correctly for all children of $z$ in $F$. Let $C$
  be the set of children of $z$ in $T$ which are not elements of
  $F$. By Lemma~\ref{lemma:set-children-single-component}, if we
  construct a tree $T_z^C$ from all the subtrees rooted at a node in
  $C$ with $z$ added as its root, the graph $G[V(T_z^C)]$ is
  connected. Since $u$ cannot be an element of this tree, it follows
  that $G[V(T_z^C)]$ is a connected subgraph of $G[V(\mathcal
  T_{X'})]$. From the computation of $T'$ defined in
  Corollary~\ref{corollary:nti-and-nice-poly} it follows that $T_z^C$
  is a subtree of $T'$. Therefore $\height{T}{x} = \height{T'}{x}$ for
  $x \in C$. The height of $z$ is $\height{T}{c} + 1$ where $c$ is a
  child of $z$ in $T$. Consider the case where $c \in C$: By
  Corollary~\ref{corollary:nti-and-nice-poly} we know that height of
  any such $c$ cannot be greater in $T'$ than in $T$, thus if
  $\height{T}{z} = \height{T}{c} + 1$ it follows that $\height{T'}{z}
  = \height{T'}{c} + 1 = h'(z)$. Otherwise, if $c \in C$, by inductive
  hypothesis $h(c) = \height{T}{c}$ which we use to set the value of
  $h(z)$.

  \paragraph{Join case}
  Finally, if $X$ is a join bag with two children $X_1$ and $X_2$ in
  $\mathcal T$, then the if-clause at line~\ref{alg:X-is-join} is
  entered. Let again $T$ be a nice treedepth decomposition rooted at
  $r$ of the graph $G[V(\mathcal T_X)]$. Then $T$ is also
  treedepth decomposition of both $G[V(\mathcal T_{X_1})]$ and
  $G[V(\mathcal T_{X_2})]$. Notice that by the properties of tree
  decompositions, $V(\mathcal T_{X_1}) \cap V(\mathcal T_{X_2}) =
  X$. Let $(F,X,h)$ be the restriction of $T$ to the set $X$ for the
  graph $G[V(\mathcal T_X)]$. Furthermore, let $T_1$ and $T_2$ be nice
  treedepth decomposition computed from $T$ by
  Corollary~\ref{corollary:nti-and-nice-poly} for the graphs
  $G[V(\mathcal T_{X_1})]$ and $G[V(\mathcal T_{X_2})]$
  respectively. Furthermore let $(F_1,X,h_1)$ be the restriction of
  $T_1$ to $X$ and $(F_2,X,h_2)$ the restriction of $T_2$ to $X$. By
  inductive hypothesis, $(F_1,X,h_1) \in R_1$ and $(F_2,X,h_2) \in
  R_2$.  Note that this is because $X = X_1 = X_2$. At some point the
  introduce operation will generate $F$ for the same reason as in the
  introduce case. By Lemma~\ref{lemma:restrictions-top-equiv} $F$ is
  both a topological generalization of $F_1$ and $F_2$. We now need to
  show that there exist two witness functions $f_1$ and $f_2$
  respectively such that the intersection of their images is exactly
  $X$.

  Let $\psi,\psi_1,\psi_2$ witness that $(F,X,h),(F_1,X,X_1),(F_2,X,h_2)$
  are restrictions of $T,T_1,T_2$ to $X$, respectively.
  By Lemma~\ref{lemma:retained-nodes} we have that $\psi_1(F_1) \subseteq \psi(F)$
  and $\psi_2(F_2) \subseteq \psi(F)$. Therefore we
  can construct $f_1 = \psi^{-1} \circ \psi_1$ and $f_2 = \psi^{-1} \circ \psi_2$,
  both of which are well-defined as $\psi$ is injective. It remains
  to show that $f_1(F_1) \cap f_2(F_2) = X$. By construction we 
  already see that $f_1(X) = f_2(X) = X$. Since $\psi_1(F_1) \subseteq T_1$
  and $\psi_2(F_2) \subseteq T_2$ with $V(T_1) \cap V(T_2) = X$, the claim follows.
  Since $f_1,f_2$ exist, the join operation will generate them at
  some point. Therefore a partial decomposition whose tree is $F$
  will be added to the result set. It remains to show that the
  height function as computed in the join operation is correct.

  Let us first show the following: let $z \in V(T)\setminus V(F)$ be
  a node whose parent is contained in $F$. Then either
  $V(T_z) \cap V(\mathcal T_{X_1}) = \emptyset$ or 
  $V(T_z) \cap V(\mathcal T_{X_2}) = \emptyset$.
  Assume to the contrary that $T_z$ contains vertices of both
  $V(\mathcal T_{X_1})$ and $V(\mathcal T_{X_2})$.
  Since $X$ separates these two sets in $G[V(\mathcal
  T_X)]$ and by assumption no vertex of $X$ is contained in $T_z$, this
  implies that $G[V(T_z)]$ has more than one connected component. But
  this contradicts $T$ being a nice treedepth decomposition. 
 
  The remaining proof parallels the proof for the introduce case. 
  Let $z$ be a leaf of $F$, $C_1$ be
  the set of children contained in $V(\mathcal T_{X_1})$ and $C_2$ the
  set of children contained in $V(\mathcal T_{X_2})$. Notice that
  since $z$ is a leaf of $F$, the set $C_1 \cup C_2$ does not contain
  any element of $X$. By
  Lemma~\ref{lemma:set-children-single-component}, the tree $T_z^{C_1}$
  induces a connected subgraph in $G[V(\mathcal T_{X_1})]$ and the 
  tree $T_z^{C_2}$ induces a connected subgraph in $G[V(\mathcal T_{X_2})]$.
  As proved earlier in the introduce case, the
  trees $T_z^{C_1},T_z^{C_2}$ are subtrees of $T_1$ and $T_2$, respectively.

  We will now show that the height function $h$ is computed correctly
  for the leaves of $F$. Let $z$ be a leaf of $F$. The height of $z$ in
  $T$ is either the height of $T_z^{C_1}$ or of $T_z^{C_2}$. By the 
  previous observation these two trees are subtrees of respectively
  $T_1$ and $T_2$, therefore their heights are given by $\height{T_1}{z}$
  and $\height{T_2}{z}$. As the height of $z$ is computed as 
  $h(z) = \max\{h_1(z), h_2(z)\}$ which is exactly $\max\{\height{T_1}{z},\height{T_2}{z}\}$,
  we conclude that the height of the leaves of $F$ is correct.
  
  We can now prove inductively that the height $h(z)$ for any
  internal node $z$ is also computed correctly. Let $C$ be the set
  of children of $z$ in $T$ which are not nodes of $F$.
  Define $C_1 = C \cap V(\mathcal T_{X_1})$ and $C_2 = C \cap V(\mathcal T_{X_2})$,
  both of which could potentially be empty.
  As previously stated, $T_z^{C_1}$ and $T_z^{C_2}$ induce connected subgraphs in
  $G[V(\mathcal T_{X_1})]$ and $G[V(\mathcal T_{X_2})]$ and are
  subtrees of $T_1$ and $T_2$, respectively.
  From Corollary~\ref{corollary:nti-and-nice-poly} we know
  that the height of $z$ in $T$ is at least the height of
  $z$ in $T_1$ (if it is contained in $T_1$) and the height
  of $z$ in $T_2$ (if it is contained in $T_2$).

  Thus it follows that if $\height{T}{z} = \height{T_z^{C_1}}{z}$,
  then $\height{T_1}{z} = \height{T_z^{C_1}}{z}$. Analogously,
  if $\height{T}{z} = \height{T_z^{C_2}}{z}$,
  then $\height{T_2}{z} = \height{T_z^{C_2}}{z}$.  
  Taking the maximum $h_1(f_1^{-1}(z))$ and $h_2(f_2^{-1}(z))$ (if the
  inverse values exist) and all the children of $z$ in $F$ therefore yields
  the correct value for $h(z)$.

  Since these are all the possible execution paths of the algorithm,
  it follows by induction that the lemma is correct.
\end{proof}

We have now shown that our algorithm will contain a partial
decomposition representing any nice treedepth decomposition of height
at most $t$. This is not sufficient to proof the correctness of the
algorithm since our tables could still contain partial decomposition
which are not restrictions of treedepth decompositions of height at
most $t$. The next lemma proofs precisely that this is not the case.

\begin{lemma}\label{lemma:all-valid-restrictions}
  Let Algorithm~\ref{fig:main-algorithm-tree-rec} be called on
  $(G,t,\mathcal{T}, X)$, where $G$ is a graph rooted at~$r$, the
  remaining parameters $t, \mathcal{T}, X$ are as described in the
  algorithm. Then every member of $R$ returned by the algorithm
  is a restriction of a treedepth decompositions of $G[V(\mathcal T_X)]$ to $X$.
\end{lemma}
\begin{proof}
  We will prove this by structural induction over tree decompositions:
  Consider the case that the tree decomposition consists of a single
  leaf bag containing only a single vertex. The returned set $R$ then 
  consists of the unique partial decomposition for this graph.

  \paragraph{Forget case}
  For the forget case, the correctness of the statement follows
  directly from Lemma~\ref{lemma:restriction-transitive} using
  the induction hypothesis.

  \paragraph{Introduce case}
  Consider the case that the bag $X$ with single child $X'$ introduces
  the vertex $u$. The set $R'$ contains, by induction hypothesis, only
  restrictions of treedepth decompositions. We have to show that
  the operation of introducing $u$ generates only restrictions of treedepth
  decompositions. Consider any $(F,X,h) \in R$.  First let us show that 
  every edge incident to $u$ in $G[V(\mathcal T_X)]$ is contained in $\clos(F)$. 
  Because $X'$ separates $u$ from $G[V(\mathcal T_{X'})\setminus X]$, any such
  edge has its other endpoint necessarily in $X'$.  Since the introduce
  operation by construction only returns restrictions with $E(G[X]) \subseteq E(\clos(F)[X])$, 
  we conclude that every edge incident to $u$ in $G[V(\mathcal T_X)]$ is contained 
  in the closure of $F$.
  
  Consider $(F',X',h') \in R'$ such that $F$ topologically generalizes $F'$ and
  further such that $(F,X,h) \in \intro_t(\{(F',X',h')\}, X', u, G)$.
  Such a restriction must, by the definition of the introduce operation,
  exist and by induction hypothesis is a restriction of a treedepth 
  decomposition $T'$ of $G[V(\mathcal T_{X'})]$.  Note that 
  every edge $vw \in E(G[V(\mathcal T_X)])$ with $v \neq u \neq w$ 
  is by induction hypothesis contained in the closure of $T'$.
  
  We will now show that we can construct a 
  treedepth decomposition $T$ of $G[V(\mathcal T_X)]$ from $T'$ of
  which $(F,X,h)$ is a restriction.
  Let $\psi'$ witness  that $(F',X',h')$ is a restriction of $T'$ to $X'$. 
  Let $f \colon V(F') \to V(F)$ be a function that witnesses that $F$ 
  topologically generalizes $F'$ with $u \notin f(F')$. We first construct
  $(\hat F,X, \hat h),(\hat F', X',\hat h')$ which are equivalent to 
  $(F,X,h),(F',X',h')$, respectively, such that $V(\hat F') \subset V(\hat F) \subseteq V(T') \cup \{u\}$
  and so that the function $f$ carried over to $\hat F', \hat F$ is simply the identity.

  By Definition~\ref{def:restriction-treedepth}, there exists $\hat F' \subseteq T'$ 
  and $\hat h'$ such that $(\hat F', X', \hat h')$ is a restriction of $T'$ to $X'$.
  Let $\hat \psi' \colon V(\hat F') \to V(F')$ be the function that witnesses 
  the equivalency of $(\hat F',X',\hat h')$ and $(F',X',h')$. Then $\hat F$ 
  is the tree with nodes $V(\hat F) = V(\hat F') \cup \{u\}$ 
  isomorphic to $F$ where the isomorphism is witnessed by the bijection
  $\phi \colon \hat F \to F$ defined via
  \begin{align*}
  \phi(v) = 
  \begin{cases}
        v  & \text{for $v = u$} \\
        f(\hat \psi'(v)) & \text{otherwise}
  \end{cases}
  \end{align*}

  and $\hat h = h \circ \phi$. We finally construct $T$ as follows:
  take the rooted forest $T'\setminus \hat F'$ and add $\hat F$ to it,
  then add the edge set $\{ xy \in E(T') \mid  x \in \hat F', y \not \in \hat F' \}$.
  
  Let us first verify that $(\hat F, X, \hat h)$, and thus by equivalency also
  $(F,X,h)$, is indeed a restriction of $T$ to $X$. By construction it is
  immediately apparent that the iterative deletion of leaves of $T$ not in $X$
  indeed yields the tree $\hat F$. However, we also need to verify that the
  height function $\hat h$ is correct, \ie that for all $v \in \hat F$,
  $\hat h(v) = \height{T}{v}$.   
  
  We prove the correctness of $\hat h$ inductively beginning at the leaves
  of $\hat F$: consider a leaf $v \in \hat F$ with $v \neq u$. The introduce operation
  calculates $h$ as $h( \phi(v)) = h'( \hat \psi'(v))$  and thus $\hat h$ as
  $\hat h(v) = \hat h'( v )$. By construction, $v$ in $T$ inherits the subtrees 
  of $v$ in $T'$,  thus $\height{T}{v} = \height{T'}{v} =\hat h(v)$. Next 
  assume $u$ is a leaf in $\hat F$: then the introduce operation sets $\hat h(v) = 1$.
  By construction  of $T$, $u$ will then not have any children and we conclude that 
  $\height{T}{u} = \hat h(u)$ in this case.  
  The  statement now follows by induction: consider any internal node $v \in \hat F$, $v \neq u$ with
  children $C$ in $T$. Let $C'$ be the set of children of $v$ in $T'$.
  By induction hypothesis, for all $w \in C \cap V(\hat F)$, $\height{T}{w} = \hat h(w)$.
  By construction of $T$ and the fact that $\hat F$ is a topological
  generalization of $\hat F'$, it holds that 
  \begin{align}
    \max_{w \in C \setminus V(\hat F)} \height{T}{w} & = \max_{w \in C' \setminus V(\hat F')} \height{T'}{w} \label{eq:first} \\
    \max_{w \in C' \cap V(\hat F')} \height{T'}{w} &\leq \max_{w \in C \cap V(\hat F)} \height{T}{w}         \label{eq:second}
  \end{align}
  Further note that
  \begin{align}
    \hat h'(v) - 1 &= \max_{w \in C'} \height{T'}{w} \nonumber \\
                     &= \max \{ \max_{w \in C'\setminus V(\hat F')} \height{T'}{w},
                                    \max_{w \in C' \cap V(\hat F')} \height{T'}{w} \} \label{eq:third}
  \end{align}
  Therefore it holds that
  \begin{align*}
    \max_{w \in C} \height{T}{w} &= \max\{ \max_{w \in C\setminus V(\hat F)} \height{T}{w},
                                    \max_{w \in C \cap V(\hat F)} \height{T}{w} \} \\
                                 &= \max\{ \max_{w \in C'\setminus V(\hat F')} \height{T'}{w}, 
                                    \max_{w \in C \cap V(\hat F)} \height{T}{w} \} && \text{by~(\ref{eq:first})}\\
                                 \begin{split}
                                 &= \max\{ \max_{w \in C'\setminus V(\hat F')} \height{T'}{w}, \max_{w \in C' \cap V(\hat F')} \height{T'}{w}, \\
                                          & \qquad   \max_{w \in C \cap V(\hat F)} \height{T}{w} \} \\
                                 \end{split}&& \text{by~(\ref{eq:second})}\\
                                 &= \max\{ \hat h'(v)-1, \max_{w \in C \cap V(\hat F)} \height{T}{w} \} && \text{by~(\ref{eq:third})} \\
                                 &= \max\{ \hat h'(v)-1, \max_{w \in C \cap V(\hat F)} \hat h(w) \} \\
                                 &= \hat h(v) - 1  && \text{\hspace*{-3cm}by introduce operation}
  \end{align*}

  The proof for $\hat h(u)$ works analogously, with the slight difference that
  $u$ will not have any children that are not in $\hat F$. We conclude 
  that $(\hat F, X, \hat h)$ and therefore $(F,X,h)$ is
  a restriction of $T$ to $X$. 
  
  It remains to show that $T$ is a treedepth
  decomposition of $G[V(\mathcal T_X)]$. Note that $V(T) = V(\mathcal T_X)$.
  By construction of $F$ and thus
  $\hat F$, edges incident to $u$ are contained in $\clos(\hat F)$
  and thus in $\clos(T)$. Since $\hat F$ is a topological generalization
  of $\hat F'$, $\clos(\hat F') \subseteq \clos(\hat F)$ and therefore
  every edge of $G[V(\mathcal T_X)]$ that lives in $V(\hat F')$ is
  contained in the closure of $T$. As $T'\setminus \hat F'$ is a subgraph of $T$, 
  the edges contained in $\clos(T'\setminus \hat F')$ are contained in $\clos(T)$.
  It remains to show that every edge $xy$ that has one endpoint $x \in \hat F'$ and the other
  endpoint $y \in T' \setminus \hat F'$ will also be covered by the closure of $T$. Consider the
  $x$-$y$-path in $T'$: this path contains a node $z \in \hat F'$ whose
  successor is not contained in $\hat F'$. Because $\hat F$ is a topological
  generalization of $\hat F'$, the node $x$ is an ancestor of $z$ in $\hat F$
  and thus in $T$. Furthermore, by construction of $T$, the node $z$ is an ancestor
  of $y$ in $T$; it follows by transitivity that $xy \in \clos(T)$.
  Therefore $T$ is a treedepth decomposition of $G[V(\mathcal T_X)]$ and 
  the lemma follows for the introduce-case.

  \paragraph{Join case}

  Consider the case of a bag $X$ with children $X_1 = X_2 = X$.
  The sets $R_1, R_2$ contain, by induction hypothesis, only
  restrictions of treedepth decompositions. We have to show that
  the operation of joining $X_1,X_2$ generates only restrictions of treedepth
  decompositions. Consider any $(F,X,h) \in R, (F_1,X,h_1) \in R_1, (F_2,X,h_2) \in R_2$ such that 
  $F$ topologically generalizes $F_1$ and $F_2$ and
  further $(F,X,h) \in \join_t(X,R_1,R_2,G)$.
  such restrictions must, by the definition of the join operation,
  exist and by induction hypothesis they are restrictions of treedepth 
  decompositions $T_1,T_2$ of $G[V(\mathcal T_{X_1})],G[V(\mathcal T_{X_2})]$, respectively.
  Note that every edge $vw \in E(G[V(\mathcal T_X)])$
  is by induction hypothesis contained either in the closure of $T_1$
  or the closure of $T_2$.
  
  We will now show that we can construct a 
  treedepth decomposition $T$ of $G[V(\mathcal T_X)]$ from $T_1,T_2$ of
  which $(F,X,h)$ is a restriction.
  
  For $i \in \{1,2\}$, let $\psi_i$ witness that $(F_i,X,h_i)$ is a restriction of $T_i$ to $X_i$. 
  Let $f_i \colon V(F_i) \to V(F)$ be a function that witnesses that $F$ topologically generalizes $F_i$.
  We first construct $(\hat F,X, \hat h),(\hat F_i, X,\hat h_i), i \in \{1,2\}$ which are equivalent to 
  $(F,X,h),(F_i,X,h_i)$, respectively, such that $V(\hat F_i) \subseteq V(\hat F) \subseteq V(T_1) \cup V(T_2)$
  and so that the functions $f_i$ that witness the topological generalization of $F_i,$ by $F$ simply
  become the identity on $\hat F_i, \hat F$.
  
  By Definition~\ref{def:restriction-treedepth}, there exists $\hat F_i \subseteq T_i$ 
  and $\hat h_i$ such that $(\hat F_i, X, \hat h_i)$ is a restriction of $T_i$ to $X_i = X$.

  Let $\hat \psi_i \colon V(\hat F_i) \to V(F_i)$ be the function that witnesses 
  the equivalency of $(\hat F_i,X,\hat h_i)$ and $(F_i,X,h_i)$. Then $\hat F$ 
  is the tree with nodes $V(\hat F) = V(\hat F_1) \cup V(\hat F_2)$ 
  isomorphic to $F$ where the isomorphism is witnessed by the bijection
  $\phi \colon \hat F \to F$ defined via
  \begin{align*}
    \phi(v) = f_i(\hat \psi_i(v)) & \qquad v \in V(\hat F_i)
  \end{align*}
  where we use the fact that for any $v \in X$, $\hat \psi_i(v) = v$ and $f_i(v) = v$.
  We further set $\hat h = h \circ \phi$. We finally construct $T$ as follows:
  take the union of the rooted forests $T_1 \setminus \hat F_1$, $T_2 \setminus \hat F_2$ and $\hat F$,
  then add for $i \in \{1,2\}$ the edge sets $\{ xy \in E(T_i) \mid  x \in \hat F_i, y \not \in \hat F_i \}$.
  
  Let us first verify that $(\hat F, X, \hat h)$, and thus by equivalency also
  $(F,X,h)$, is indeed a restriction of $T$ to $X$. By construction it is
  immediately apparent that the iterative deletion of leaves of $T$ not in $X$
  indeed yields the tree $\hat F$. However, we also need to verify that the
  height function $\hat h$ is correct, \ie that for all $v \in \hat F$,
  $\hat h(v) = \height{T}{v}$. 

  We prove the correctness of $\hat h$ inductively beginning at the leaves
  of $\hat F$: consider a leaf $v \in \hat F$. Since $v \in X$, the join operation
  calculates $h$ as $h(v) = \max_{i \in \{1,2\}} h_i(v)$ and thus $\hat h$ as
  $\hat h(v) = \max_{i \in \{1,2\}} h_i(v)$. By construction, $v$ in $T$ inherits the subtrees 
  of $v$ in $T_1$ and of $v$ in $T_2$,  thus $\height{T}{v} = \max_{i \in \{1,2\}} \height{T_i}{v} =\hat h(v)$. 
  The  statement now follows by induction: consider any internal node $v \in \hat F$ with
  children $C$ in $T$. For $i \in \{1,2\}$, let $C_i$ be the set of children of $v$ in $T_i$.
  By induction hypothesis, for all $w \in C \cap V(\hat F)$, $\height{T}{w} = \hat h(w)$.
  By construction of $T$ and the fact that $\hat F$ is a topological
  generalization of $\hat F_1, \hat F_2$, it holds that 

  \begin{align}
    \max_{w \in C \setminus V(\hat F)} \height{T}{w} & = \max_{i \in \{1,2\}} \max_{w \in C_i \setminus V(\hat F_i)} \height{T_i}{w} \label{eq:firstJ} \\
    \max_{i \in \{1,2\}} \max_{w \in C_i \cap V(\hat F_i)} \height{T_i}{w} &\leq \max_{w \in C \cap V(\hat F)} \height{T}{w}         \label{eq:secondJ}
  \end{align}
  Further note that
  \begin{align}
    \hat h_i(v) - 1 &= \max_{w \in C_i} \height{T_i}{w} \nonumber \\
                     &= \max \{ \max_{w \in C_i\setminus V(\hat F_i)} \height{T_i}{w},
                                    \max_{w \in C_i \cap V(\hat F_i)} \height{T_i}{w} \} \label{eq:thirdJ}
  \end{align}
  Therefore it holds that
  \begin{align*}
    \max_{w \in C} \height{T}{w} &= \max\{ \max_{w \in C\setminus V(\hat F)} \height{T}{w},
                                    \max_{w \in C \cap V(\hat F)} \height{T}{w} \} \\
                                 &\hspace*{-1.5cm}= \max\{ 
                                   \max_{i \in \{1,2\}} \max_{w \in C_i \setminus V(\hat F_i)} \height{T_i}{w},  
                                   \max_{w \in C \cap V(\hat F)} \height{T}{w} \} && \text{by~(\ref{eq:firstJ})}\\
                                 \begin{split}
                                 &\hspace*{-1.5cm}= \max\{ \max_{i \in \{1,2\}} \{ \max_{w \in C_i \setminus V(\hat F_i)} \height{T_i}{w},
                                           \max_{w \in C_i \cap V(\hat F_i)} \height{T_i}{w} \}
                                          \\
                                          & \qquad   \max_{w \in C \cap V(\hat F)} \height{T}{w} \} \\
                                 \end{split}&& \text{by~(\ref{eq:secondJ})}\\
                                 &\hspace*{-1.5cm}= \max\{ \max_{i \in \{1,2\}} \hat h_i(v)-1, \max_{w \in C \cap V(\hat F)} \height{T}{w} \} && \text{by~(\ref{eq:thirdJ})} \\
                                 &\hspace*{-1.5cm}= \max\{ \max_{i \in \{1,2\}} \hat h_i(v)-1, \max_{w \in C \cap V(\hat F)} \hat h(w) \} \\
                                 &\hspace*{-1.5cm}= \hat h(v) - 1   && \text{\hspace*{-2cm}by join operation}
  \end{align*}

  It remains to show that $T$ is a treedepth
  decomposition of $G[V(\mathcal T_X)]$. Note that $V(T) = V(\mathcal T_X)$.
  Since $\hat F$ is a topological generalization
  of $\hat F_i$ for $i \in \{1,2\}$, it holds that $\clos(\hat F_i) \subseteq \clos(\hat F)$ and therefore
  every edge of $G[V(\mathcal T_X)]$ that lives in $V(\hat F_i)$ is
  contained in the closure of $T$. As $T_i \setminus \hat F_i$ is by
  construction a subgraph of $T$, the edges contained in each 
  $\clos(T_i \setminus \hat F_i)$ are contained in $\clos(T)$.
  It remains to show that for $i \in \{1,2\}$, every edge $xy$ that has 
  one endpoint $x \in \hat F_i$ and the other endpoint 
  $y \in T_i \setminus \hat F_i$ will also be covered by the closure of $T$.
  Consider the $x$-$y$-path in $T_i$: this path contains a node $z \in \hat F_i$ whose
  successor is not contained in $\hat F_i$. Because $\hat F$ is a topological
  generalization of $\hat F_i$, the node $x$ is an ancestor of $z$ in $\hat F$
  and thus in $T$. Furthermore, by construction of $T$, the node $z$ is an ancestor
  of $y$ in $T$; it follows by transitivity that $xy \in \clos(T)$.
  Therefore $T$ is a treedepth decomposition of $G[V(\mathcal T_X)]$ and 
  the lemma follows for the introduce-case.
\end{proof}

\begin{lemma}\label{lemma:algorithm-correctness}
  Algorithm~\ref{fig:main-algorithm-tree} decides the
  treedepth of the input graph $G'$.
\end{lemma}
\begin{proof}
  First, it is easy to see that $\mathcal T$ is a nice tree decomposition of
  $G$ of width $t+1$. By Lemma~\ref{lemma:all-restrictions-contained}
  and Lemma~\ref{lemma:all-valid-restrictions} it follows that the set
  $R$ contains all restrictions of any nice treedepth decomposition
  rooted at $r$ of the rooted graph $G$ after
  line~\ref{alg:top-rec-call} of
  Algorithm~\ref{fig:main-algorithm-tree-rec} is executed. From
  Corollary~\ref{corollary:restriction-structure} we know that the
  height of the partial decomposition equals the height of the
  treedepth decomposition of which it is a restriction. From
  Lemma~\ref{lemma:all-valid-restrictions} we know that every partial
  decomposition in $R$ is a restriction of a treedepth
  decomposition of $G$. From Lemma~\ref{lemma:nice-td-exists} and
  Lemma~\ref{lemma:root-always-root} we know that there is a nice
  treedepth decomposition rooted at $r$ of minimal height of the
  rooted graph $G$. From Lemma~\ref{lemma:depth-rooted-graph} we know
  that $G$ has a treedepth decomposition of height $t+1$ if and only
  if $G'$ has one of height $t$. Thus the return statement at
  line~\ref{alg:final-return} will give the correct answer.
\end{proof}

\subsection{Running time of Algorithm~\ref{fig:main-algorithm-tree}}
\label{sec:runtime-alg-tw}

\begin{lemma}\label{lemma:count-restrictions}
  For a set $X$, the number of possible restrictions on $X$
  of height at most $t$ is, up to equivalency, bounded 
  by $2^{|X| t+ |X| \log t + |X| \log |X|}$.
\end{lemma}
\begin{proof}
  For any restriction $(F,X,h)$ of height at most $t$,
  we have that $|F| \leq |X|\cdot t$,  since every leaf of $F$
  is contained in $X$ and $height(F) \leq t$.

  First note that any monotone path $P$ (\ie a path on which every
  node is either an ancestor or a descendant of any other node on
  the path) inside the forest of a restriction can be labeled by
  $h$ in at most $2^t$ ways: since $h$ will increases strictly while
  following $P$ from top to bottom and the $|P| \leq t$, the function
  $h |_P$ is already completely determined by the set $h(P)$.

  Consider any ordering $x_1,\dots,x_{|X|}$ of the elements in $X$ 
  and denote by $X_i = \{x_1,\dots,x_i\}$, for $1 \leq i \leq |X|$.
  We upper-bound the number of restrictions by considering the following
  construction: given a restriction $(F,X_i,h_i)$, we have
  at most $i\cdot t \cdot 2^t$ ways of constructing a restriction 
  on $X_{i+1}$: we choose one of $i \cdot t$ nodes of $F$ and 
  attach one of the possible $2^t$ labeled paths to it, with leaf-node 
  $x_{i+1}$. We allow adding a path of length zero, this operation simply
  exchanges the initially chosen node with $x_{i+1}$.
  Clearly all restrictions on $X_{i+1}$ can be generated in such a
  way from restrictions on $X_i$. Thus the number of restrictions is
  given by
  \begin{align*}
    \prod_{i = 1}^{|X|} ti 2^{t} &= 2^{|X|t + |X| \log t} |X|! \leq 2^{|X| t + |X| \log t + |X| \log |X|}.
  \end{align*}
\end{proof}

\begin{lemma}\label{lemma:count-topgens}
  Given restrictions $(F,X,h),(F',X',h')$ with $X' \subseteq X$ there are at most
  $2^{t \cdot |X'|/2}$ ways of how $F$ can topologically generalize $F'$ and in this
  time, all candidate maps witnessing this fact can be generated.
\end{lemma}
\begin{proof}
  We upper bound the number of possible maps $f$ that witness
  that $F$ is a topological generalization of $F'$. Consider a leaf node $v \in F'$,
  which necessarily is contained in $v \in X' \subseteq X$. Let $P_v'$ be the path 
  from the root of $F'$ to $v$ (in $F'$) and $P_v$ the path from the root of $F$ to $v$
  (in $F$). In order for $f$ to preserve the ancestor relationship of 
  vertices in $F'$, the vertices of $P_v'$ must be mapped to vertices of $P_v$ while preserving
  order, \ie if $x$ appears before $y$ in $P_v'$ then $f(x)$ must appear before $f(y)$ in $P_v$.
  It follows that there are exactly ${|P_v| \choose |P_v'|}$ ways of how $f$ could map $P_v'$ to $P_v$.

  We now upper bound the number of maps by taking the product of all such paths:
  \begin{align*}
    \prod_{v \in X'} {|P_v| \choose |P_v'|} &\leq 2^{t \cdot |X'|/2}
  \end{align*}
  using the fact that no rooted path in $F$ and $F'$ exceeds length $t$. This method
  can be used constructively (since we can check whether a map indeed witnesses a
  topological generalization in polynomial time) to enumerate all maps.
\end{proof}

\begin{lemma}\label{lemma:runnint-time-algo-2}
  Algorithm~\ref{fig:main-algorithm-tree-rec} called on $G$, $t$,
  $\mathcal T$ and $X$, where $G$ is a graph rooted at $r$ of size
  $n$, $\mathcal T$ is a nice tree decomposition of $G$ of width
  $w$ where every bag contains $r$ and $X$ is a bag of $\mathcal T$
  runs in time $O(2^{4wt + 3w \log wt} \cdot wt \cdot n)$.
\end{lemma}
\begin{proof}
  A nice tree decomposition can only have $O(n)$ bags, therefore the
  linear dependence follows easily.

  By Lemma~\ref{lemma:count-restrictions}, the set $R$ of restrictions
  at any given time cannot contain more then $2^{wt + w \log t + w
    \log w}$ elements. During the join case, we generate all possible
  restrictions $(F,X,h)$ and for each we consider all pairs
  $(F_1,X,h_1),(F_2,X,h_2)$ from the respective tables $R_1,R_2$ of
  the child bags. For such a pair we need to compute all possible maps
  $f_1,f_2$ that might witness the fact that $F$ topologically
  generalizes both $F_1$ and $F_2$. To check if a function witnesses a
  topological generalization takes linear time in the size of the
  trees, \ie~$O(t \cdot |X|)$.  The total amount of time needed for this
  operation, using the bound provided by
  Lemma~\ref{lemma:count-topgens}, is at most
  \begin{align*}
    (2^{wt + w \log t + w \log w})^3 \cdot (2^{t/2 \cdot w})^2 \cdot O(wt) & = O(2^{4wt + 3w \log wt}\cdot wt)  \\
  \end{align*}

  Both forget- and introduce-operation and checking if the result set
  already contains an equivalent partial decomposition have running
  times bounded by this function, thus $O(2^{4wt + 3w \log wt}\cdot
  wt)$ is also an upper bound for the total running time of every
  operation and the lemma follows.
\end{proof}

We finally are able to sum up the results in the following theorem, a direct consequence 
of Lemma~\ref{lemma:algorithm-correctness} and Lemma~\ref{lemma:runnint-time-algo-2}.

\begin{theorem}\label{theorem:td}
  Let $G$ be a graph of size $n$ and $t$ an integer. Given a tree
  decomposition of $G$ of width $w$, one can decide in time and space
  $O(2^{4wt + 3w \log wt} \cdot wt \cdot n)$ whether $G$ has treedepth at most $t$ and
  if so, output a treedepth decomposition of that height.
\end{theorem}

To actually construct a solution, we keep the tables of all bags in
memory and employ backtracking to reconstruct a minimal treedepth
decomposition.


\section{A Simple Algorithm}\label{sec:Simple}
We can now use Theorem~\ref{theorem:td} to answer
the problem posed by Ossona de Mendez and \Nesetril
in \cite{NOdM12}:

\begin{prob}
  Is there a simple linear time algorithm to check $td(G) \leq t$ for
  fixed $t$? Is there a simple linear time algorithm to compute a
  rooted forest $Y$ of height $t$ such that $G \subseteq clos(Y)$
  (provided that such a rooted forest exists)?
\end{prob}

The problem is motivated by the fact that treedepth---being a minor-closed
property---can be expressed in monadic second order logic and thus
one can employ Courcelle's theorem \cite{Cou90} to compute the treedepth
of a graph of bounded treewidth in linear time. The above problem is
motivated by the fact that the running time of this approach is 
unclear. More specifically, the standard proof of Courcelle's involves 
creating a tree-automaton whose size cannot be bounded
by any elementary function in the formula size unless P=\NP \cite{FG04}. 
Here we will show that we can
use the algorithm presented in Chapter~\ref{sec:Algorithm} to give a
much more direct and simpler algorithm.

\begin{proposition}[\cite{NOdM12}]
  Let $G$ be a graph of treedepth $t$. Then a treedepth decomposition
  which is the tree given by a depth first search of $G$ is a
  $2^t$-approximation of the treedepth.
\end{proposition}

It is easy to see that the tree given by a depth first search of a
graph gives a treedepth decomposition of the graph. As a path
of length $2^t$ would witness the fact that the treedepth is larger
than $t$, one either obtains a DFS of height most $2^t$ or can
correctly conclude that the given graph has treedepth $> t$.

\begin{proposition}[\cite{NOdM12}]
  Let $G$ be a graph and let $T$ be DFS-tree of $G$.
  Then there exists a path decomposition of $G$ whose width 
  is the height of $T$. This path decomposition can be computed in linear time.
\end{proposition}

The algorithm we presented in Chapter~\ref{sec:Algorithm} expected a
tree decomposition as part of its input. Thanks to the above two
lemmas we can construct an algorithm which only takes a graph as its input, 
\cf Algorithm~\ref{fig:simple-algorithm}.

\begin{algorithm}
  \small
  \caption{treedepth-simple{\label{fig:simple-algorithm}\sc }}
  \KwIn{A graph $G$, an integer $t$}
  \KwOut{Is the treedepth of $G$ smaller or equal to $t$?}  \BlankLine

  Start computing a tree $Y$ representing a depth first search in $G$\;
  \While{Computing $Y$}
  {
    \If{Depth is greater than $2^t$}
    {
      \Return No\;
    }
  }

  Compute a nice path decomposition $\mathcal P$ of $G$ from $Y$\;
  \Return treedepth-on-tree-decomposition($G,t,\mathcal P$)\;
\end{algorithm}

The following theorem now follows from Theorem~\ref{theorem:td} and
Algorithm~\ref{fig:simple-algorithm}.

\begin{theorem}
  There is a simple algorithm to decide whether the
  treedepth of a graph is at most $t$ in time and space $2^{2^{O(t)}} n$
  and, in the positive case, output a treedepth decomposition witnessing
  this fact.
\end{theorem}

We point out that Algorithm~\ref{fig:simple-algorithm} can be made to run
in logarithmic space.

\begin{lemma}
  The algorithm in Algorithm~\ref{fig:simple-algorithm} can be made to
  run in logarithmic space for a fixed treedepth $t$.
\end{lemma}
\begin{proof}
  Sketch: It is easy to see that a depth first search can
  be implemented in such a way that only the current path from the
  root of the search tree to the leaf must be kept in memory. Since
  the depth of our search is bounded by the constant $2^t$ we can
  compute the search tree in logarithmic space. The contents of each
  bag in the path decomposition are precisely the paths from the root
  to a leaf of the search tree. This means that we can compute the
  bags of the path decomposition in parallel to the computation of the
  search tree. Since for path decompositions we only need the forget
  and introduce operations, we only need to keep two bags in memory at
  any point. Since the size of the tables is bounded by a function in
  $t$, it follows that the algorithm in
  Algorithm~\ref{fig:simple-algorithm} can be implemented in such a
  way that it only uses logarithmic space.
\end{proof}

All in all, we consider our algorithm to solve the stated
problem.


\section{Fast Algorithm}\label{sec:Fast}
We can also extend the algorithm in Chapter~\ref{sec:Algorithm} to get
a fast algorithm. The simple version we presented in the previous
chapter might be simple, but it runs in double exponential time. This
is because to get a tree decomposition\footnote{actually a path
  decomposition} whose width is bounded in $t$ we use a depth first
search tree, which is very easy to use and does not require any
complex mathematical tools, but only gives us a tree decompositions
whose width is $2^t$. If we could bound the width of the tree
decomposition we compute before running
Algorithm~\ref{fig:main-algorithm-tree} linearly in $t$ then we could
get a much better running time. We want to show this is actually
easily possible using known results.

First notice that $tw(G) \leq pw(G) \leq td(G)$ for any graph
$G$. This fact was implicitly proven in Chapter~\ref{sec:Simple} since
we showed how it easy to construct a path decomposition of width $t$
from a treedepth decomposition of height $t$. It follows then that if
we bound the width of the tree decomposition linearly on the treewidth
of the graph we will also be bounding it on the treedepth of
graph. There is a recent results which proofs there exists an
algorithm which runs in time $2^{O(t)}n$ and calculates a
5-approximation tree decomposition for a graph \cite{BDDFLP13}. The
algorithm in Algorithm~\ref{fig:fast-algorithm} shows an algorithm
that uses these two facts to give a fast algorithm.

\begin{algorithm}
  \small
  \caption{treedepth-fast{\label{fig:fast-algorithm}\sc }}
  \KwIn{A graph $G$, an integer $t$}
  \KwOut{Is the treedepth of $G$ smaller or equal to $t$?}  \BlankLine

  Compute a $5$-approximated nice tree decomposition $\mathcal T$ of
  $G$ using the algorithm in \cite{BDDFLP13}\;

  \If{No such tree decomposition is found}
  {
    \Return No\label{alg:return-no-aprox}\;
  }

  \Return treedepth-on-tree-decomposition($G,t,\mathcal T$)\;
\end{algorithm}

\begin{lemma}
  Algorithm~\ref{fig:fast-algorithm} decides the treedepth of the
  input graph $G$.
\end{lemma}
\begin{proof}
  Since $tw(G) \leq td(G)$ if the graph $G$ has treedepth $t$ then
  there must exist a tree decomposition of width at most $5t$ which is
  a $5$-approximation for the treewidth of the graph. Thus returning
  with a negative result on line~\ref{alg:return-no-aprox} is
  correct. From Theorem~\ref{theorem:td} we know that the
  call to Algorithm~\ref{fig:main-algorithm-tree} decides if $G$ has
  treedepth $t$, thus the lemma follows.
\end{proof}

\begin{lemma}
  Algorithm~\ref{fig:fast-algorithm} runs in time $2^{O(t^2)}$
  on the input graph $G$ and integer $t$.
\end{lemma}
\begin{proof}
  Since $tw(G) \leq td(G)$ of follows that the width of the tree
  decomposition $\mathcal T$ is at most $5t$. From \cite{BDDFLP13} we
  know that the running time of computing the tree decomposition is
  $2^{O(t)}n$ and from Theorem~\ref{theorem:td} the the call to
  Algorithm~\ref{fig:main-algorithm-tree} is $2^{O(tw)} n$, where $w$
  is the width of $\mathcal T$. Since $w \leq 5t$ it follows that the
  running time of the call to Algorithm~\ref{fig:main-algorithm-tree}
  is $2^{O(t \cdot 5t)} \cdot n = 2^{O(t^2)} \cdot n$.
\end{proof}

For the same reason as in previous chapter, backtracking can also be
used in this algorithm to compute an actual treedepth decomposition of
height $t$ or less. Thus the following theorem follows.

\begin{theorem}
  Let $G$ be a graph of size $n$. Deciding if $G$ has a treedepth
  decomposition of height $t$ and constructing such a treedepth
  decomposition can be computed in time $2^{O(t^2)} n$.
\end{theorem}

To the best of our knowledge the best running time to compute the
exact treewidth $w$ of a graph takes time $2^{O(w^3)} n$, which means
that our algorithm computes treedepth of a graph faster than the best
known algorithm for treewidth computes the treewidth of a graph.


\section{Treedepth and Chordal graphs}\label{sec:Implications}
As mentioned in the introduction, deciding treedepth remains NP-hard
even on chordal graphs. Interestingly, the special structure of tree
decompositions of chordal graphs can be used to reduce the running time of our 
algorithm significantly with only minor changes. To the best of our knowledge, 
no such algorithm was known so far (an algorithm with exponential dependence 
on the number of cliques in a chordal graph was presented in~\cite{AH94}).
Since obtaining an optimal tree decomposition
for chordal graphs is possible in linear time, we do not need the
treewidth approximation here.

\begin{lemma}
  Given a chordal graph $G$ and an integer $t$, one can decide in time and space 
  $2^{O(t \log t)} \cdot n$ whether $\td(G) \leq t$ and in the positive case output a
  treedepth decomposition of that height.
\end{lemma}
\begin{proof}
  Since adding a universal vertex to a chordal graph does not violate the
  chordality, we tacitly assume in the following that such a vertex $r$ exists.
  First, check whether $\omega(G) > t$ and if that is the case, output that the
  treedepth of $G$ is greater than $t$. Otherwise, $\omega(G) \leq t$ which implies
  that $\tw(G) \leq t$. Compute a clique tree of $G$ in linear time (\cf~\cite{BP93}),
  \ie a tree decomposition of $G$ in which every bag induces a clique.

  If we now run Algorithm~\ref{fig:main-algorithm-tree-rec} on $G$ we can show 
  that only partial decompositions whose forest is a path are kept during
  each step of the dynamic programming: consider a bag $X$ and a set of restrictions $R$
  computed by the algorithm. For any $(F,X,h) \in R$, the condition 
  $E(G[X]) \subseteq E(\clos(F)[X])$ must be fulfilled
  (in the join- and introduce-case this is explicitly enforced and it is easy to see that 
  the forget-case cannot create a non-path from a path). Therefore, all elements of $X$
  lie in a single path from the root to a leaf in $F$---but since in a restriction every
  leaf of $F$ must be a member of $X$, this path is exactly $F$. 
  The maximum number of restrictions of height at most $t$ and
  whose forest is a single path is bounded by $2^{O(t \log t)}$, 
  \cf proof of Lemma~\ref{lemma:count-restrictions}. If we modify the introduce- and
  join-procedure of Algorithm~\ref{fig:main-algorithm-tree-rec} to only generate restrictions
  whose forests are paths, which by the previous observation are the only restrictions that
  would be kept in any case, the running time reduces to the claimed bound.
\end{proof}

\section{Conclusions and Further Research}\label{sec:Conclusion}
We provide an explicit simple self contained algorithm, \ie~an
algorithm which does not rely on any other complex results, which for
a fixed $t$ decides if a graph $G$ has treedepth $t$ or computes a
treedepth decomposition of height $t$ if one exists in linear
time. This answers an open question posed in~\cite{NOdM12}. We also
provide an explicit algorithm to decide the treedepth or construct a
minimal treedepth decomposition of a given graph in time
$2^{O(t^2)}n$.

A natural question that arises is whether one can find a
constant-factor approximation for treedepth in single-exponential
time, similar to the algorithm for treewidth. Such an algorithm would
be interesting in the sense that it would remove the dependency of the
algorithm provided in this paper from the treewidth-approximation
(hoping that a direct approximation of treedepth would be simpler).

On the topic of width-measures, it still remains open whether graphs of low treedepth
admit fast algorithms that are impossible on graphs of low pathwidth. This is 
motivated further by the fact that the construction proving lower bounds on graphs of 
bounded pathwidth clearly contain very long paths and thus have high
treedepth~\cite{LMS11a}. Proving similar bounds for graphs of bounded treedepth 
would be equally insightful.


\def\redefineme{
    \def\LNCS{LNCS}%
    \def\ICALP##1{Proc. of ##1 ICALP}%
    \def\FOCS##1{Proc. of ##1 FOCS}%
    \def\COCOON##1{Proc. of ##1 COCOON}%
    \def\SODA##1{Proc. of ##1 SODA}%
    \def\SWAT##1{Proc. of ##1 SWAT}%
    \def\IWPEC##1{Proc. of ##1 IWPEC}%
    \def\IWOCA##1{Proc. of ##1 IWOCA}%
    \def\ISAAC##1{Proc. of ##1 ISAAC}%
    \def\STACS##1{Proc. of ##1 STACS}%
    \def\IWOCA##1{Proc. of ##1 IWOCA}%
    \def\ESA##1{Proc. of ##1 ESA}%
    \def\WG##1{Proc. of ##1 WG}%
    \def\LIPIcs##1{LIPIcs}%
    \def\LIPIcs{LIPIcs}%
    \def\LICS##1{Proc. of ##1 LICS}%
}

\bibliographystyle{abbrv}
\bibliography{cross,conf}

\clearpage


\end{document}